\documentclass[journal]{IEEEtran}

\ifCLASSOPTIONcompsoc
 \usepackage[caption=false,font=normalsize,labelfont=sf,textfont=sf]{subfig}
\else
 \usepackage[caption=false,font=footnotesize]{subfig}
\fi

\usepackage{url}

\usepackage{amsfonts}       
\usepackage{nicefrac}       
\usepackage{microtype}      
\usepackage{lipsum}
\usepackage{graphicx}
\graphicspath{ {./images/} }

\usepackage{float}
\usepackage{algorithm}
\usepackage[noend]{algpseudocode}

\usepackage[monochrome]{xcolor}

\usepackage{amsmath}
\usepackage{amssymb}
\usepackage{amsthm}
\usepackage{amsfonts}
\usepackage{nccmath}

\newtheorem{theorem}{Theorem}
\newtheorem{lemma}[theorem]{Lemma}
\newtheorem{proposition}[theorem]{Proposition}
\newtheorem{corollary}[theorem]{Corollary}

\theoremstyle{definition}
\newtheorem{definition}{Definition}

\theoremstyle{remark}
\newtheorem{remark}{Remark}
\newtheorem{example}{Example}

\usepackage{mathtools}
\usepackage{float}
\usepackage[noadjust]{cite}
\usepackage{balance}
\usepackage[noabbrev]{cleveref}
\crefformat{equation}{(#1)}
\usepackage{bbm}
\usepackage{bm}
\usepackage{resizegather}
\usepackage[super]{nth}

\usepackage{acronym}
\newacro{DGR}{Data-Guided Regulation}
\newacro{F-DGR}{Fast Data-Guided Regulation}
\newacro{SVD}{Singular Value Decomposition}
\newacro{PBH}{Popov-Belevitch-Hautus}
\newacro{LQR}{Linear Quadratic Regulator}
\newacro{PI}{Policy Iteration}
\newacro{RLS}{Recursive Least Squares}
\newacro{sysID}{system identification}
\newacro{LMI}{Linear Matrix Inequalities}
\newacro{LTI}{Linear Time-Invariant}
\newacro{MPC}{Model Predictive Control}

\usepackage{cite}
\usepackage{textcomp}
\newcommand\BibTeX{{\rm B\kern-.05em{\sc i\kern-.025em b}\kern-.08em
		T\kern-.1667em\lower.7ex\hbox{E}\kern-.125emX}}

\usepackage{enumerate}
\newcommand{\RomanNumeralCaps}[1]
    {\MakeUppercase{\romannumeral #1}}

\newcommand\x{{\bm x}}
\def\u{{\bm u}}
\newcommand\z{{\bm z}}
\newcommand\w{{\bm w}}
\def\v{{\bm v}}

\newcommand \e{{\bm e}}
\newcommand\oomega{{\bm \omega}}
\newcommand\ogamma{{\bm \gamma}}

\newcommand\p{{\bm p}}
\newcommand\q{{\bm q}}
\newcommand\ozeta{{\bm \zeta}}

\newcommand\scalemath[2]{\scalebox{#1}{\mbox{\ensuremath{\displaystyle #2}}}}

\usepackage{fancyhdr}
\fancypagestyle{firststyle}
{
   \fancyhf{}
   \chead{\scriptsize This article has been accepted for publication in IEEE Transactions on Automatic Control. This is the author's version which has not been fully edited and content may change prior to final publication. Citation information: DOI 10.1109/TAC.2021.3131148\vspace{0.3cm}}
   \cfoot{\scriptsize \copyright~2021 IEEE. Personal use is permitted, but republication/redistribution requires IEEE permission. See https://www.ieee.org/publications/rights/index.html for more information.}
   
}
\begin{document}
\title{
On Regularizability and its Application to Online Control of Unstable LTI Systems}

%
%
%

\author{Shahriar~Talebi, ~\IEEEmembership{Student~Member,~IEEE,}
Siavash~Alemzadeh, ~\IEEEmembership{Student~Member,~IEEE,}\\
Niyousha ~Rahimi, ~\IEEEmembership{Student~Member,~IEEE,}
and~Mehran~Mesbahi,~\IEEEmembership{Fellow,~IEEE}
\thanks{The research has been supported by AFOSR grant FA9550-20-1-0053. The authors are with the William E. Boeing Department of Aeronautics and Astronautics, University of Washington, Seattle, WA, USA. Emails: \tt\small \{shahriar, alems, nrahimi, mesbahi\}@uw.edu.}
\thanks{A preliminary version of this paper has appeared in the \textit{\nth{59} IEEE Conference on Decision and Control~\cite{talebi2020online}}.}}

\markboth{}%
{Talebi \MakeLowercase{\textit{et al.}}: On Regularizability and its Application to Online Control of Unstable LTI Systems}
%



\maketitle
\thispagestyle{firststyle}
\begin{abstract}
		%
		Learning, say through direct policy updates, often requires assumptions such as knowing \emph{a priori} that the initial policy (gain) is stabilizing, {\color{PineGreen}or persistently exciting (PE) input-output data, is available.}
		In this paper, we examine online regulation of (possibly unstable) partially unknown linear systems with no {\color{PineGreen} prior access to an initial stabilizing controller nor PE input-output data;
		we instead leverage the knowledge of the input matrix for online regulation}.
		First, we introduce and characterize the notion of ``regularizability'' for linear systems that gauges the extent by which a system can be regulated in finite-time in contrast to its asymptotic behavior (commonly characterized by stabilizability/controllability).
		Next, having access only to the input matrix, we propose the \ac{DGR} synthesis procedure that---as its name suggests---regulates the underlying state while also generating informative data that can subsequently be used for data-driven stabilization or system identification.
		We further improve the computational performance of DGR via a rank-one update and demonstrate its utility in online regulation of the X-29 aircraft.
		%
		%
\end{abstract}


%
\IEEEpeerreviewmaketitle

\section{Introduction}
	
		
	
	
	\IEEEPARstart{F}{eedback} control is ubiquitous in modern technology including applications where it 
	provides means of stabilization in addition to performance.
	Control of open-loop unstable plants arising for instance, in industrial and flight control applications, underscores the importance of stabilization with robustness guarantees.
    As such, control of unstable systems is an ongoing research topic, particularly in the context of safety-critical systems.
	It is well-known that unstable systems are fundamentally more difficult to control~\cite{stein2003respect}; in fact, practical closed-loop systems with unstable subsystems are only locally stable \cite{sree2006control}.
	Yet, most of the existing synthesis literature has focused on model-based control where the designer has to discern fundamental limitations stemming from process instabilities~\cite{skogestad2002control}.
	
	Recent interest in model-free stabilization in the meantime, has been motivated by novel sensing technologies, robust machine learning, and efficient computational methods to reason about control and estimation of uncertain systems--all from measured (online) data \cite{hou2017data,sedghi2020multi}.
	Safety-critical systems have in fact necessitated non-asymptotic analysis on data-driven methods~\cite{faradonbeh2018finite, dean2019safely}.
	In particular, there has been a growing interest in examining
	finite-time control of unknown linear dynamical systems from time-series or a single trajectory \cite{alaeddini2018linear, sarkar2019finite, berberich2019robust, oymak2019non, fattahi2019learning, wagenmaker2020active, tsiamis2019finite}.
	Parallel to asymptotic analysis in traditional adaptive control and \ac{sysID} \cite{ljung2001system, narendra2012stable, aastrom2013adaptive}, model-based finite-time control has benefited from a least-squares approach to identification followed by robust synthesis--see for example~\cite{dean2017sample}.
	In this direction, probabilistic bounds on the estimation error related to the required run-time have been examined.
	While it has been shown that model-based methods require fewer measurements for certain control problems in general~\cite{tu2018gap},\footnote{That is, first finding a model estimate from data and then use that estimate for control design.} data collection required for \ac{sysID} can be expensive or impractical due to resource limitations and safety constraints.
	Furthermore, some of the aforementioned studies rely on \textit{a priori} information about the system, such as estimates of system parameters \cite{dean2019safely}, an initial stabilizing controller \cite{kim2005stable, bradtke1994adaptive, fazel2018global, alemzadeh2019distributed}, or assuming an open-loop stable system \cite{oymak2019non, sarkar2019finite}.

	{\color{PineGreen}
It is known that an input-output trajectory of a controllable linear time invariant (LTI) system can be parameterized by (offline) data trajectories generated from a persistently exciting (PE) input~\cite{willems2005note}. Building on this fact, there has been recent works on stabilization of 
LTI systems directly from the available data (e.g., see \cite{de2019formulas, coulson2019deepc, baros2020online, van2019data, yu2021controllability, berberich2021data}).
	However, ensuring a PE input-output data may not be practical for data-driven control or identification of unstable systems even in low dimensions without recourse to resets~\cite{de2019formulas}.\footnote{For instance, injecting white noise into an unstable system can result in ill-conditioned data matrices, that in turn, leads to numerical issues.}
    Hence, existing data-guided methods might not be directly applicable for safety critical control such as
    online flight control \cite{lozano2004robust} or infrastructure recovery \cite{gonzalez2017efficient}.
    Our work is motivated by such applications, requiring no reliance on an initial stabilizing controller nor a PE input-output data trajectory for data-guided control. 
	In this direction, we focus on instances where 
	the input matrix of the LTI system is known.
	This point of view has been adopted by the desire to ensure satisfactory performance for online data-guided control based on a single trajectory--even when the underlying system is unstable--from the onset.
    }
	
	In order to realize the above program in a systematic way, in the first part of the paper, we introduce a class of linear systems exhibiting a property called ``regularizability;''\footnote{Not to be confused with the notion of ``regularity" for singular systems~\cite{ozcaldiran1990regularizability}.} this notion captures the input ``effectiveness" as it relates to finite-time regulation.
	We then proceed to characterize regularizability using linear matrix inequalities (LMIs), as well as clarify  how it relates to spectral properties of the underlying LTI system. 
	Additionally, we show how this system-theoretic notion can be verified in a more transparent manner
    for a subclass of partially known systems.

	In the second part of this paper, by employing the notion of regularizability, we introduce the \acf{DGR} algorithm, an online iterative synthesis procedure that utilizes a single trajectory for an otherwise partially unknown (discrete time) LTI system.
    {\color{PineGreen}\ac{DGR} does not use prior assumptions on the linear state dynamics, nor access to an initial stabilizing  controller or an input-output dataset; instead, the algorithm only relies on the knowledge of the input matrix.
	The knowledge of the input matrix is motivated by scenarios where it is known how the control input affects the state dynamics, yet how the internal states of the system interact is uncertain (for example, consider the problem of controlling an unknown networked system from a given set of nodes).
	This assumption also proves useful for our setup in order to, 1) ensure a satisfactory performance for the system trajectory from the onset of the regulation process, 2) avoid requiring an initial stabilizing controller, and 3) avoid requiring a PE input-output trajectory from an unstable systems (that is often impractical and leads to ill-conditioned data matrices for post-processing). We postulate that
	in the case when the input matrix is also unknown, deriving nontrivial guarantees for closed loop performance of unstable systems from the onset might prove to be illusive.
	Finally, as pointed out above, having access to the input structure of a system is pertinent to a number of applications that involve learning~\cite{vrabie2009adaptive, wagenmaker2020active, nozari2017network, sharf2018network}; a similar assumption has been adopted for learning and control of nonlinear systems \cite{jagtap2020control}, where the system dynamics is affine in control with known input mapping and unknown state dynamics.}
	
	The contribution of the proposed work is as follows:
	(1) in addition to introducing the notion of regularizability for LTI systems, we show how it is distinct from related properties such as stabilizability.
	We believe that regularizability is of independent interest particularly as {\color{violet} it pertains to online regulation};
	(2) we derive conditions under which \ac{DGR} can eliminate unstable modes of the (unknown) system and regulate its state trajectory.\footnote{Here, regulation is ensured by bounding the norm of the system states during the learning process; see \Cref{sec:probSetup} for more details.}
	DGR essentially aims for simultaneous identification and regulation of the hidden unstable modes from a single trajectory in a feedback form.
	As such, \ac{DGR} can avoid some of the conditioning issues
	that arise in processing data generated by an unstable system.
	Using the notion of regularizability, we then proceed to derive upper bounds on the state trajectories based on a geometric quantity for LTI systems that we refer to as the ``instability number;'' 
	(3) we show that while \ac{DGR} performs well for a large class of unstable systems, special structures (e.g., symmetry) further facilitate deriving intuitive bounds on the system trajectory during the learning process;
	(4) finally, we show that the discrete nature of time-series data enables a recursive approach to DGR synthesis.
	%
	In this direction a recursive DGR is proposed that circumvents storing the entire data history and avoids demanding operations such as pseudoinverse computation or multiplying large matrices.\footnote{	
A preliminary version of this work is the manuscript~\cite{talebi2020online}. The contributions of the present work as compared with~\cite{talebi2020online} include various LMI characterizations of regularizability, its extension to polytopic uncertain systems, its use in the context of online regulation with a more general setup involving the control cost, recursive and efficient updates of the online regulation algorithm, as well as a more detailed discussion on the examples and relevant literature. Furthermore, the proofs and analysis that are not presented in the conference version of the paper have been included in this manuscript.}

	The rest of the paper is organized as follows.
	In \S\ref{sec:math}, we provide an overview of mathematical notions 
	used in the paper.
	In \S\ref{sec:probSetup}, we introduce the problem setup as well as a motivating example, followed
	by introducing the notion of regularizability for an LTI model.
	We further study the properties of regularizable systems~in \S\ref{sec:regularizability}.
	Additionally, the \ac{DGR} algorithm is proposed in \S\ref{sec:DGR} as the
 means of online regulation of (possibly) unstable systems.
	The subsequent part of \S\ref{sec:DGR} is devoted to the analysis of the {DGR-induced} closed loop system, deriving upper bounds on the state trajectories, and efficient implementation of \ac{DGR}.
	We provide an illustrative example in \S\ref{sec:simulation} followed by
	concluding remarks in \S\ref{sec:conclusion}.

\section{Mathematical Preliminaries}
\label{sec:math}
    We denote the fields of real and complex numbers by $\mathbb{R}$ and $\mathbb{C}$, respectively, and real $n\times m$ matrices by $\mathbb{R}^{n \times m}$.
    %
	%
	%
	The $n \times 1$ vector of all ones is denoted by $\mathbbm{1}$.
	The unit vector $\e_i$ is a column vector with identity at its $i$th entry and zero elsewhere.
	The $n \times n$ identity matrix is denoted by $\mathrm{I}_n$ (or simply $\mathrm{I}$).
	The $\mathrm{diag}(.)$ indicates a diagonal matrix constructed by elements of its argument in the same order starting from upper-left corner.
	For a real symmetric matrix $L$, we say that $L\succ 0$ when $L$ is positive-definite (PD) and $L\succeq 0$ for the 
 	positive-semidefinite (PSD) case.
    The {algebraic multiplicity} of an eigenvalue $\lambda$ is denoted by $m(\lambda)$;
$\lambda$ is called {simple} if $m(\lambda)=1$.
	The range and nullspace of a real matrix $M \in \mathbb{R}^{n \times m}$ are denoted by $\mathcal{R}(M) \subseteq \mathbb{R}^n$ and $\mathcal{N}(M)\subseteq \mathbb{R}^m$, respectively, the dimension of $\mathcal{R}(M)$ is designated by $\mathrm{rank} (M)$, and its transpose by $M^\intercal$.
	The dimension of a vector space is denoted by $\textbf{dim}$. 
	The span of a set of vectors over the complex field is denoted by
	$\mathrm{span}\{.\}$.
	The singular value decomposition of a matrix $M \in \mathbb{R}^{n\times m}$ is the factorization $M= {U} \Sigma {V}^\intercal$, where the unitary matrices ${U}\in \mathbb{R}^{n\times n}$ and ${V}\in \mathbb{R}^{m\times m}$ consist of the left and right ``singular" vectors of $M$, and $\Sigma \in \mathbb{R}^{n\times m}$ is the diagonal matrix of singular values in a descending order.
	The reduced order matrices $U_r,V_r$ can be obtained by truncating the factored matrices $U$ and $V$ in the SVD to the first $r$ columns, where $r = \mathrm{rank}(M)$. The thin SVD of $M$ is then the factorization $M = U_r \Sigma_r V_r^\intercal$, where $\Sigma_r \in \mathbb{R}^{r \times r}$ is now nonsingular.
    From SVD, one can also construct the Moore-Penrose generalized inverse ---\textit{pseudoinverse} for short--- of $M$ as $M^{\dagger}={V}\Sigma^{\dagger}{U}^\intercal$, in which $\Sigma^{\dagger}$ is obtained from $\Sigma$ by first replacing each nonzero singular value with its inverse (zero singular values remain intact) followed by a transpose.
	A square matrix $A \in \mathbb{R}^{n\times n}$ is Schur stable if $\rho(A) < 1$, where $\rho(.)$ denotes the spectral radius, i.e., maximum modulus of eigenvalues of its matrix argument. The matrix $A$ is (complex)~{diagonalizable} if there exist a diagonal matrix $\Lambda \in \mathbb{C}^{n\times n}$ and a nonsingular matrix $U \in \mathbb{C}^{n\times n}$ such that $A = U \Lambda U^{-1}$.
	In this case, $\Lambda$ consists of the eigenvalues of $A$ with columns of $U$ as the corresponding eigenvectors. 
    The orthogonal projection of a vector $\bm v$ on a linear subspace $S$ is denoted by $~\Pi_{S}(\bm v)$.\footnote{We will be working with finite dimensional vector spaces and as such all subspaces are closed.}
    When the columns of a matrix $U\in \mathbb{R}^{n \times k}$ form an orthonormal basis for the subspace $S$, then $~\Pi_{S} = U U^\intercal$.
	The {Euclidean norm} of a vector $\x \in\mathbb{R}^n$ is denoted by $\| \x \|=(\x^\intercal \x)^{1/2}$.
	For a matrix $M$, its {operator norm} is denoted by $\|M\|=\sup\{\| M {\u} \|: \|{\u} \|=1\}$.
	By $\mathcal{B}_{2}^r$, we refer to the $r$-dimensional Euclidean ball of unit radius.
	An $r$-dimensional multi-index $\alpha$ is an $r$-tuple of the form $(\alpha_1, \alpha_2, \cdots, \alpha_r)$ with all non-negative integers $\alpha_i$, where the sum of its elements is denoted by $|\alpha| = \sum_{i=1}^r \alpha_i$;
 $\alpha \in \{0,1\}^r$ signifies that each $\alpha_i\in\{0,1\}$ for $i=1,\dots,r$.
	We say that \textit{$\x_0$ excites $k$ modes of a matrix $A$} if $\x_0$ is contained in the (complex-)span of $k$ eigenvectors of $A$, but not in the span of any $k-1$ eigenvectors;
	we refer to those $k$ eigenvectors (for which $\x_0$ is in the span of) as the corresponding \textit{excited modes}.
\section{Problem Setup}
\label{sec:probSetup}

	In this section, we introduce the problem setup and
	highlight its unique features through an example.
	Consider a discrete-time LTI model of the form,
	\begin{equation}
    	\label{eqn:sys-dynamic}
    	\x_{t+1} = A \x_t + B \u_t, \hspace{10mm} \x_0\ \; \text{given},
	\end{equation}
	where $A\in\mathbb{R}^{n\times n}$ and $B\in\mathbb{R}^{n\times m}$ are the system parameters and $\x_t\in\mathbb{R}^{n}$ and $\u_t\in\mathbb{R}^{m}$ denote the state and control inputs at 
	time index $t$, respectively.
	We assume that the system matrix $A$ is unknown 
	and (possibly) unstable, and that the input matrix $B$ is known.
	The problem of interest is to design $\u_t$ from online state measurements (and not the system matrix $A$ nor the offline data) such that: \RomanNumeralCaps{1}) the system is regulated, with a norm uniformly bounded during the learning process, e.g., $\x_t$ evolves in a (safe) region with a quantifiable bounded norm, and the corresponding data matrix does not become ill-conditioned, and \RomanNumeralCaps{2}) the system generates informative data for post-processing, for example in the context of data-driven stabilization or system identification.\footnote{We interchangeably use the terms linear independence and \textit{informativity} of data to emphasize that the collected data has useful information content for decision-making; the orthogonal ``hidden" signal $\z_t$ in \Cref{lem:algo-dyn} further exemplifies this perspective.}
	
	{\color{PineGreen}
    Considering regulation by having access to the input matrix is of interest in applications where it is known a priori how various control inputs effect the dynamic states, e.g., how the elevator deflection effects the aircraft pitch dynamics, or  influencing a diffusive network from certain boundary nodes.
    %
    Intuitively, this assumption allows an online regulation mechanism to have a chance of stabilizing an unknown (and possibly) unstable system in real-time from the onset of the learning process.}
        
    The following example motivates our setup and underscores why the data-guided perspective requires introducing new system theoretic notions.
	\begin{example}
		\label{ex:motivateExample}
		For any positive integer $n$, define the system matrix $A \in \mathbb{R}^{n\times n}$ and the input matrix $B \in \mathbb{R}^{n}$ as,
		\begin{gather*}
		    A = \scalemath{0.8}{
		\left(
		\begin{array}{*5{c}}
		\lambda_1   & 1         &    0      &  \hdots   &     0     \\
		0           & \lambda_2 &    1      &           &  \vdots   \\
		0           & 0         & \lambda_3 &  \ddots   &     0     \\
		\vdots      &           &  \ddots   &  \ddots   &     1     \\
		0           &  \hdots   &           &     0     & \lambda_n \\
		\end{array}
		\right)}, \quad 
		B = \scalemath{0.8}{\left(
		\begin{array}{*1{c}}
		0\\
		~\\
		\vdots\\
		~\\
		0\\
		1
		\end{array}\right)}.
		\end{gather*}
		Note that for any choice of $\lambda_i\in\mathbb{R}$, the pair $(A,B)$ is controllable (and therefore stabilizable).
		Furthermore, since the set $\{\lambda_i\}$ coincides with the spectrum of $A$, if any subset of $\{\lambda_i\}$ are equal, then $A$ contains the corresponding Jordan block.
		Moreover, when $\lambda_i\neq \lambda_j$ ($i\neq j$), then $A$ is diagonalizable.
		Let $\x_0 = \e_1$ and observe that under \eqref{eqn:sys-dynamic}, we have $\e_1^\intercal \x_t = \lambda_1^t$ for all $0 \leq t < n$ regardless of the input $\u_t$.
		This implies that, for ``any'' choice of input, for the first $n$ iterations, the first state of the system grows exponentially fast with the rate $\lambda_1$ whenever $|\lambda_1|> 1$.
	\end{example}{}

	\begin{remark}
	\label{rmk:motivateExample}
	\Cref{ex:motivateExample} constructs a family of controllable systems where no controller can regulate their respective first states--at least for the first $n$ iterations.
	That is, a system state will grow exponentially fast regardless of the choice of $\u_t$, even when all eigenvalues of $A$ except $\lambda_1$ are stable (e.g., $|\lambda_i| < 1$ for $i=2,\cdots,n$).
	Note that in this example, the (right) eigenvector associated with the unstable mode of $A$ (i.e., the eigen-pair $(\lambda_1,\e_1)$) is orthogonal to $\mathcal{R}(B) = \mathcal{R}(\e_n)$.
	This is despite the fact that the \ac{PBH} controllability test holds (i.e., for any left eigenvector $\v$ of $A$ we have $\v^\intercal B \neq 0$).
	This example highlights that controllability of a pair $(A,B)$ does not capture ``regularizability'' of an unstable linear system, specially when closed loop regulation has to be achieved in a data-guided manner and from the onset of the learning process.
	%
	Finally, we point out that in the particular case when $\lambda_i = 0$ for $i=2,\dots,n$, the controllability matrix corresponding to $(A,B)$ is anti-diagonal with all anti-diagonal elements equal to identity.
	Therefore, it has singular values/eigenvalues all equal to $\pm 1$.
	This implies that the controllability matrix has condition number equal to identity; as such modes that are difficult to
	regularize are not distinguished by the controllability matrix.
	\end{remark}

	In order to formalize the behavior of the class
	of systems mentioned above, we introduce a system theoretic notion that captures the effectiveness of the input as pertinent to online regulation.
    In order to motivate this notion, note that the dynamics in \cref{eqn:sys-dynamic} can be represented as,
    \begin{align*}
        \x_{t+1} 
        &= ~\Pi_{\mathcal{R}(B)^\perp} A \x_t + ~\Pi_{\mathcal{R}(B)} A \x_t + B \u_t\\
        &= ~\Pi_{\mathcal{R}(B)^\perp} A \x_t + B ( B^\dagger A \x_t + \u_t).
    \end{align*}
    Setting $\u_t = -B^\dagger A \x_t + \Bar{\u}_t$, \cref{eqn:sys-dynamic} can be rewritten as 
    \(\x_{t+1} = \widetilde{A} \x_t + B \Bar{\u}_t \)
    where,
    \begin{equation}
        \label{eqn:Atilde}
        \widetilde{A} := ~\Pi_{\mathcal{R}(B)^\perp} A,
    \end{equation}
    and $\Bar{\u}_t$ is yet to be designed.
    Note that the signals $\widetilde{A} \x_t$ and $B \Bar{\u}_t$ are now orthogonal. 
    This implies that the control signal would not directly affect the part of dynamics that is generated by $~\Pi_{\mathcal{R}(B)^\perp} A$. 
    As such, in order to have even the possibility of achieving ``some" online performance for this system in finite-time, we require that this part of the dynamics be stable. 
	This observation thereby motivates
	the following definition.
	
	\begin{definition}
		\label[definition]{def:Regularizable}
		The pair $(A,B)$ is called \textit{regularizable} if $\widetilde{A}: = ~\Pi_{\mathcal{R}(B)^\perp} A~$ is Schur stable.
	\end{definition}
	
	As we will show subsequently, regularizability of a pair $(A,B)$ is related to the stabilizability of $(A,B)$ as well as detectability of $(A, B^\intercal)$; a combination that is not typically
	encountered in LTI analysis.
	This connection is intuitive, as regulation of a system in finite-time requires the states to be accessible (for control and observation) through the input matrix $B$.
    Regularizability also facilitates a new perspective on LTI systems, providing a basis for the analysis of online algorithms such as the one proposed in \Cref{sec:DGR}.

\section{Regularizable Systems}
\label{sec:regularizability}

    In order to get a better sense of the notion of regularizability, we study the spectral properties of $\widetilde{A}$ in \cref{eqn:Atilde} and its relation with
    system matrices $A$ and $B$. 
    First, the following example highlights why regularizability of a system is
    distinct from its controllability.
	\begin{example}
		\label{ex:regular-vs-stable}
    	Consider the linear system with $A$ defined as in \Cref{ex:motivateExample} such that $|\lambda_1|>1$ and $|\lambda_i|<1$ for $i=2,\dots,n$.
    	Note that the pair $(A,\e_n)$ is controllable (and thus stabilizable); however this pair is not  regularizable.
    	On the other hand, the pair $(A,\e_1)$ is  regularizable but not controllable.
	\end{example}
	
	Recall that a pair $(A,B)$ is stabilizable if and only if $(A^\intercal, B^\intercal)$ is detectable. The detectability of $(A, B^\intercal)$ is seldom of interest in linear system theory \cite{Hespanha2018linear}; however, we show that it is indeed, a necessary condition for $(A,B)$ to be regularizable.
	To this end, we first connect regularizability to the spectral properties of the pair $(A,B)$.
	\begin{lemma}
		\label[lemma]{lem:eigenpair}
		Let $\widetilde{A} = ~\Pi_{\mathcal{R}(B)^\perp} A$.
		Then for each right eigenpair $(\lambda,\v)$ of $A$ the following holds:
		\begin{itemize}
			\item $(\lambda,\v)$ is a right eigenpair of $\widetilde{A}$ whenever $\v \in \mathcal{R}(B)^\perp$ or $\lambda = 0$.
			
			\item $(0,\v)$ is a right eigenpair of $\widetilde{A}$ whenever $\v \in \mathcal{R}(B)$.
		\end{itemize}
	\end{lemma}
	The proof of \Cref{lem:eigenpair} directly 
	follows from the definitions and therefore is omitted.
	Note that the above lemma does not address the scenario where $(\lambda, \v)$ is an eigenpair of $A$, with $\lambda \neq 0$, and $\v=\v_1+\v_2$, 
	with nontrivial $\v_1 \in \mathcal{R}(B)$ and $\v_2 \in \mathcal{R}(B)^\perp$.
	%
	%
	The following example illustrates that $\widetilde{A}$, as a product of matrix $A$ with an orthogonal projection operator, has a spectral radius distinct from $A$.
	\begin{example}
    	Consider the system in Example~\ref{ex:motivateExample}, where the identity off-diagonal elements of $A$ are replaced with $10$, $\lambda_1 = 0.9$ and $\lambda_i = 0$ for all $i=2,\dots, n$, and $B = \mathbbm{1}$.
    	It is straightforward to show that for all $n\geq 2$, $A$ is Schur stable with spectral radius of $0.9$ while $\widetilde{A}$ is not, i.e., $(A,B)$ is not regularizable.
    	In this case, in spite of $A$ being Schur stable, its operator norm is about $10$.
    	Furthermore, the spectral radius of $\widetilde{A}$ would be $4.55$ for $n=2$ and increases to about $10$ as $n$ increases.
    	This results in a pathological behavior despite the fact that the system is originally stable, e.g.,
    	any infinite horizon closed-loop \ac{LQR} controller for this system would demonstrate undesirable behavior ---similar to \Cref{ex:motivateExample}--- when initialized from $\x_0 = \mathbbm{1}$.\footnote{One practical remedy to this problem is to split the dynamics into multiple time-scales using, say, a sampling heuristics \cite{manohar2019optimized}.
    	However, time-scale separation often requires physical insights and non-trivial to identify for  general systems~\cite{naidu2001singular}, let alone for a system with an unknown dynamics.}.
    	%
    	%
    	Finally, it is worth noting that the controllability matrix of this pair is ill-conditioned in contrast to \Cref{ex:motivateExample}.
	\end{example}
	
	The preceding discussion exemplifies that even for a stable system, it is nontrival to assert that state trajectories over a finite time horizon are ``well-regulated.''
	It is no surprise then that, in spite of its severe limitations from a system theoretic perspective, most of the recent works on data-guided control focus on {\em contractible} systems as they streamline composition rules and analysis for consecutive iterations in a learning algorithm~\cite{lale2020regret, agarwal2019online}. However, the succeeding remark shows why regularizability, as introduced in this work, is less restrictive, and thus---by replacing contractility---can mitigate those system theoretic limitations. 
	%
	%
	
	\begin{remark}
	A pair $(A,B)$ is said to be contractible if there exists a controller $K$ such that $\|A - B K\| < 1$.
	Noting that
	\[A- B K = \widetilde{A} + ~\Pi_{\mathcal{R}(B)}(A - B K),\]
	for any vector $\x \in \mathbb{R}^n$, (by orthogonality)
	it follows that,
	\begin{align*}
	    \|\widetilde{A} \x\|^2  
	    &= \|(A-BK) \x\|^2 - \|~\Pi_{\mathcal{R}(B)} (A-BK)\x\|^2\\
	    &\leq \|(A-BK)\|^2 \, \|\x\|^2.
	\end{align*}
	This, in turn, implies that a contractible system is regularizable (as in that case $\|\widetilde{A}\|< 1$).
	In particular, if the original system matrix $A$ is non-expansive (at least on the subspace $A^{-1}\{\mathcal{R}(B)^\perp\}$), then $(A,B)$ is regularizable.
	\end{remark}
	

    The following results further clarifies the relation between regularizable systems and their system theoretic twins.
	
	\begin{proposition}
	   \label[proposition]{prop:dirReg-vs-stab}
    	If $(A,B)$ is  regularizable, then
    	\begin{itemize}
    	    \item $(A,B)$ is stabilizable, and
    	    \item $(A,B^\intercal)$ is detectable.
    	\end{itemize}{}
	\end{proposition}
	\begin{proof}
	    For the first claim, note that $\widetilde{A} = A - ~\Pi_{\mathcal{R}(B)} A = A + B K$, where $K:= -B^\dagger A$. Thus if $(A,B)$ is  regularizable then $K$ is a stabilizing closed loop controller.
	    For the second claim, we establish a contrapositive.
    	Suppose that $(A,B^\intercal)$ is not detectable. Hence there exists a right eigenpair $(\lambda, \v)$ of $A$, where $|\lambda|\geq 1$ and $\v\in\mathcal{N}(B^\intercal)=\mathcal{R}(B)^{\perp}$.
    	Then, \Cref{lem:eigenpair} implies that $(\lambda, \v)$ must be a right eigenpair of $\widetilde{A}$.
    	Since $|\lambda|\geq 1$, $\widetilde{A}$ is not Schur stable and therefore $(A,B)$ is not  regularizable.
	\end{proof}
    
    Note that the consequents of \Cref{prop:dirReg-vs-stab} are equivalent whenever $A$ is symmetric, as detectability of $(A,B^\intercal)$ is equivalent to stabilizability of $(A^\intercal,B)$.
    Also, note that  \Cref{prop:dirReg-vs-stab} provides a necessary condition for regularizability, whereas the following counter-example underscores why the stabilizability of $(A,B)$, even when combined with detectability of $(A,B^\intercal)$, is not sufficient.
    
    \begin{example}
        Let the system matrices $A,B$ be defined as in Example \ref{ex:motivateExample} and consider the pair $(A_1,B_1) := (A+A^\intercal, B)$.
        By the structure of $A_1$, note that $(A_1,B_1)$ is controllable.
        Since $A_1$ is symmetric, $(A_1,B_1^\intercal)$ is also observable. By direct computation we observe that,
        \begin{gather*}
             \widetilde{A} = ~\Pi_{\mathcal{R}(B)^\perp} A = \scalemath{0.8}{
    		\left(
    		\begin{array}{*5{c}}
    		2 \lambda_1    &     1    &    0      &  \hdots   &  0        \\
    		1           & 2 \lambda_2 &    1      &  \ddots          &  \vdots    \\
    		0           & \ddots         & \ddots &   \ddots   &  0        \\
    		 \vdots     &      \ddots     & 1   &  2 \lambda_{n-1}   &  1        \\
    		0           &  \hdots   &    0       &   0       & 0 \\
    		\end{array}
    		\right)}.
        \end{gather*}
    		Now if any of $\lambda_i$'s, for $i=1, \dots, n-1$, is say, larger than $1/2$, then $\widetilde{A}$ would be unstable, implying that $(A_1,B_1)$ is not  regularizable.
    \end{example}
    
    In order to complete our understanding of regularizability, we provide several characterizations using \ac{LMI}.
    
    \begin{proposition}\label[proposition]{prop:LMI-charac}
    Consider a pair $(A,B)$, and denote $~\Pi_\perp : = ~\Pi_{\mathcal{R}(B)^\perp}$.
    Then the following are equivalent:
    \begin{enumerate}[(i)]
        \item The pair $(A,B)$ is  regularizable.
        \item $\exists P \succ 0$ such that $\rho(A^\intercal ~\Pi_\perp P ~\Pi_\perp A P^{-1}) < 1$.
        \item $\exists P \succ 0$ such that $\|P^{1/2} ~\Pi_\perp A P^{-1/2}\| < 1$.
        \item $\exists P \succ 0$ such that 
        \(A^\intercal ~\Pi_\perp P ~\Pi_\perp A - P \prec 0.\)
        \item $\exists W \succ 0$ such that  
        \(\scalemath{0.9}{\left(\begin{array}{cc}
	     W& ~\Pi_\perp A W \\
	     W A^\intercal ~\Pi_\perp & W
	\end{array}{}\right)} \succ 0.\)
	    \item $\exists P\succ 0$ and $G \in \mathbb{R}^{n\times n}$ such that,
	    \[\scalemath{0.9}{\left(\begin{array}{cc}
	     P&  A^\intercal ~\Pi_\perp G^\intercal \\
	     G ~\Pi_\perp A & G + G^\intercal -P
	\end{array}{}\right)} \succ 0.\]
	    \item $\exists P \succ 0$, and $G, H \in \mathbb{R}^{n\times n}$ such that, \[\scalemath{0.9}{\left(\begin{array}{cc}
	     G A + A^\intercal G^\intercal - P &  A^\intercal H^\intercal - G \\
	     H A - G^\intercal & ~\Pi_\perp P ~\Pi_\perp - H -H^\intercal
	\end{array}{}\right)} \prec 0.\]
    \end{enumerate}
    \end{proposition}
    \begin{proof}
        Noting that regularizability of $(A,B)$ is equivalent to Schur stability of $~\Pi_\perp A$, the first four equivalences are direct consequences of Theorem 7.7.7 in \cite{horn2012matrix}.
        By using Schur complements and constructing a congruence induced by $\text{diag}(I,P^{-1})$, {(iv)} becomes equivalent to {(v)}.
        The last two equivalences are due to Theorem 1 in \cite{Oliveiraa1999AND} and Theorem 1 in \cite{Oliveira1999LMICO}, respectively.
    \end{proof}
    
    We conclude this section by providing a sufficient condition for
    guaranteeing when a polytopic uncertain LTI system is regularizable.
    
    \begin{proposition}
        \label[proposition]{prop:LMI-robustness}
        Consider $A_i \in \mathbb{R}^{n\times n}$ for $i=1,\dots,N$ and suppose there exist matrices $P_i \succ 0$ and $G,H \in \mathbb{R}^{n\times n}$ satisfying,
    	\[\scalemath{0.9}{\left(\begin{array}{cc}
    	     G A_i + A_i^\intercal G^\intercal - P_i &  A_i^\intercal H^\intercal - G \\
    	     H A_i - G^\intercal & ~\Pi_{S} P_i ~\Pi_{S} -H-H^\intercal
    	\end{array}{}\right)} \prec 0,\]
    	for some linear subspace $S \subseteq \mathbb{R}^n$.
    	Then a pair $(A, B)$ is regularizable whenever $A \in \mathrm{convhull}\{A_i\}_1^N$ and
    	\[\scalemath{0.9}{\left(\begin{array}{cc}
    	     P_i &  P_i ~\Pi_{\mathcal{R}(B)^\perp} \\
    	     ~\Pi_{\mathcal{R}(B)^\perp} P_i  & ~\Pi_{S} P_i ~\Pi_{S}
    	\end{array}{}\right)} \succeq 0, \quad \forall i=1,\dots, N.\]
    \end{proposition}
    
    \begin{proof} 
    Since $A \in \mathrm{convhull}\{A_i\}_1^N$, there exists scalars $ \alpha_i \in [0,1]$ with $\sum_1^N \alpha_i = 1$ such that $A = \sum_1^N \alpha_i A_i$.
    By defining  $P = \sum_1^N \alpha_i P_i$ and taking the convex combinations  of the negative definite matrices in the hypothesis with weights $\alpha_i$ we obtain,
    \begin{gather}
    \scalemath{0.9}{
    \left(\begin{array}{cc}
	     G A + A^\intercal G^\intercal - P &  A^\intercal H^\intercal - G \\
	     H A - G^\intercal & ~\Pi_{S} P ~\Pi_{S} - H -H^\intercal
    	\end{array}{}\right)} \prec 0.     
	    \label{eqn:LMI-in-proof}
    \end{gather}
	Now by taking the Schur complement of the \ac{LMI} in the hypothesis involving the input matrix $B$ it follows that,
	\[~\Pi_{S} P_i ~\Pi_{S} \succeq ~\Pi_{\mathcal{R}(B)^\perp} P_i ~\Pi_{\mathcal{R}(B)^\perp}, \quad \forall i=1,\dots, N.\]
Convex combinations of these \ac{LMI}s with the same coefficients lead to,
	\(~\Pi_{S} P ~\Pi_{S} \succeq ~\Pi_{\mathcal{R}(B)^\perp} P ~\Pi_{\mathcal{R}(B)^\perp}.\)
	This, together with the \ac{LMI} in \cref{eqn:LMI-in-proof} imply the \ac{LMI} in \Cref{prop:LMI-charac}.\textit{(vii)}. 
	As $P \succ 0$, we conclude that the pair $(A,B)$ is regularizable.
    \end{proof}

\begin{remark}
\color{Blue}
    Note that the proof above also shows that the last \ac{LMI} in the statement of \Cref{prop:LMI-robustness} is equivalent to 
    \begin{equation}\label{eqn:LMI-monotone}
        ~\Pi_{S} P_i ~\Pi_{S} \succeq ~\Pi_{\mathcal{R}(B)^\perp} P_i ~\Pi_{\mathcal{R}(B)^\perp}, \quad \forall i=1,\dots, N;
    \end{equation}
    which is certainly satisfied when $S = \mathcal{R}(B)^\perp$.
    Thus, a direct consequence of \Cref{prop:LMI-robustness}--together with the characterization in \Cref{prop:LMI-charac}.(vii)--is as follows: if there exists an input matrix $B$ such that $(A_i,B)$ is regularizable for each $i=1,\dots, N$, then we can conclude that $(A, B)$ is regularizable for any (unknown) matrix $A \in \mathrm{convhull}\{A_i\}_1^N$. 
    This observation does not follow directly from the definition as spectral radius is not subadditive.
    Moreover, \Cref{prop:LMI-robustness} provides the flexibility of working with the linear subspace $S$ independently of $\mathcal{R}(B)$, which proves
    to be useful for design purposes, e.g., devising an input matrix in order to make a polytopic uncertain system regularizable.
\end{remark}
{\color{Blue}
Finally, \Cref{prop:LMI-robustness}--in view of \cref{eqn:LMI-monotone}--implies that 
regularizability is a monotonic system theoretic property
with respect to the input, in the sense that enlarging $\mathcal{R}(B)$ would
not destroy its regularizability. In fact, ``enlarging'' $\mathcal{R}(B)$ for a system would make it ``more'' regularizable (as $\rho(\widetilde{A})$ will be smaller).}

\section{\acf{DGR} Algorithm}
\label{sec:DGR}

The primary focus of this section is devising an online, data-driven feedback controller to regulate the system's state trajectories, quantified in terms
of a signal norm.
In this direction, we propose an iterative procedure for updating the feedback gain (policy);
the form of the controller can be motivated by considering, at each iteration $t$,
the following optimization problem with a ``one-step quadratic cost'',\footnote{The setup resembles dead-beat control design, with the caveat that the synthesis is data-guided.}

\begin{equation}
	\begin{aligned}	\label{eqn:optimization}
    	&\textstyle \min_{\u_t} \hspace{3mm} \|\x_{t+1}\|^2 + \alpha \|\u_t\|^2  \\
    	&\hspace{2mm} \text{s.t.} \hspace{5mm} \x_{t+1} = A \x_t + B \u_t,
	\end{aligned}
\end{equation}
where $\x_t$ is measured over time but the system matrix $A$ is unknown, and $\alpha \geq 0$ is a regularization factor for the controller design.\footnote{\color{Blue} We note that considering a more elaborate form of cost (e.g., finite/infinite horizon \ac{LQR} cost) for this optimization problem is certainly relevant. However, in this specific problem setup, i.e., no prior knowledge on the matrix $A$ and absence of any prior input-state data, we have observed no significant numerical advantage in considering a more elaborate cost--particularly for upper bounding the state trajectories from the onset of the learning process.}
%
%
In the case of known $A$, it is straightforward to characterize the set of minimizers of the above optimization problem through the first order optimality condition,
\begin{equation*}
	(\alpha I_m + B^\intercal B) \u_t + B^\intercal A \x_t = 0;
\end{equation*}
as such, the corresponding input belongs to a linear subspace in $\mathbb{R}^m$ parameterized by the system matrices and data.
{\color{violet} The following proposition illustrates why regularizability as presented in \S\ref{sec:regularizability} is pertinent to online regulation of LTI systems.
\begin{proposition}
    For every $\alpha \in [0, \varepsilon)$, with some small enough $\varepsilon>0$, the minimum norm solution of the iterative optimization \cref{eqn:optimization} stabilizes the system \cref{eqn:sys-dynamic} if and only if the pair $(A,B)$ is regularizable.
\end{proposition}
\begin{proof}
Given a fixed $\alpha \geq 0$, the minimum norm solution to \cref{eqn:optimization} at iteration $t$ is $\u_t^* = - G_\alpha A \x_t$, where $G_\alpha \coloneqq (\alpha I + B^\intercal B)^\dagger B^\intercal$. Therefore, this iterative solution stabilizes the system in  \cref{eqn:sys-dynamic} if and only if $A - B G_\alpha A$ is Schur stable. Using properties of the pseudoinverse, $A - B G_0 A = (I - B B^\dagger (B B^\dagger)^\intercal) A = (I - B B^\dagger)A = \widetilde{A}$, where $\widetilde{A}$ is as defined in \Cref{def:Regularizable}. The proof now follows by continuity of the spectral radius with respect to $\alpha$.
\end{proof}

Note that $G_\alpha \to 0$ as $\alpha \to \infty$, implying that $\u_t \to 0$ for all $t$. As such, in general, the solution to \cref{eqn:optimization} is stabilizing when $\alpha$ is small enough.}
The formulation of the optimization problem \cref{eqn:optimization} requires the knowledge of system parameters; nonetheless, it forms the basis for the proposed algorithm when $A$ is unknown and potentially unstable. 
%
The corresponding synthesis procedure is detailed in \Cref{alg:Controller}.
Specifically, for any $\alpha \geq 0$, at iteration $t$, \ac{DGR} sets
\begin{equation}
    \label{eqn:opt-input}
    \u_t^*= -K_t^* \x_t, \hspace{7mm} K_t^* \coloneqq G_\alpha \mathcal{Y}_{t} \mathcal{X}_{t-1}^{\dagger},
\end{equation}
{\color{violet}
where $G_\alpha \coloneqq (\alpha I + B^\intercal B)^\dagger B^\intercal$ and $\mathcal{X}_{t-1}, \mathcal{Y}_{t} \in \mathbb{R}^{n \times t}$ are 
the measured data matrices,
\begin{align*}
\mathcal{X}_{t-1} &\coloneqq \begin{pmatrix} \x_0 &\dots & \x_{t-1}\end{pmatrix},\\
\mathcal{Y}_{t} &\coloneqq \begin{pmatrix} \x_1 -B\u_0 &\dots & \x_t -B\u_{t-1}\end{pmatrix}.
\end{align*}}
%
%
\begin{algorithm}[!t]
	\caption{\acf{DGR}}
	\begin{algorithmic}[1]
		\State \textbf{Initialization} \hspace{1mm} (at $t=0$)
		\State \hspace{5mm} Measure $\x_0$; set $K_0 = \mathbf{0}$, {$G_\alpha = (\alpha I + B^\intercal B)^\dagger B^\intercal$}
		\State \hspace{5mm} Set $\mathcal{X}_0 = \left(\begin{array}{c} \x_0 \end{array}\right)$ and $\mathcal{Y}_0 = \left(\begin{array}{c} \ \end{array}\right)$
		\State \textbf{While stopping criterion not met}\footnotemark
		\State \hspace{5mm} Compute \hspace{1mm} $\u_t = -K_t \x_t$
		\State \hspace{5mm} Run system \eqref{eqn:sys-dynamic} and measure $\x_{t+1}$
		\State \hspace{5mm} Update \hspace{1mm} $\mathcal{Y}_{t+1} =
	    \left(\begin{array}{cc}
		\mathcal{Y}_t & \x_{t+1} - B \u_t
		\end{array}\right)$
		\State \hspace{18mm} $K_{t+1} = G_\alpha \mathcal{Y}_{t+1} \mathcal{X}_t^{\dagger}$
		\State \hspace{18.5mm} $\mathcal{X}_{t+1} = \left(\begin{array}{cc}
			\mathcal{X}_t & \x_{t+1}
		\end{array}\right)$
		\State \hspace{5mm} $t = t+1$
	\end{algorithmic}
	\label{alg:Controller}
\end{algorithm}
\footnotetext{The stopping criterion can be application specific.
For instance, for \ac{sysID} generating $n$ linearly independent data is sufficient, while mere stabilization may require less; see~\cite{van2019data}.}
Intuitively, collecting more data results in capturing the
essential (e.g., unstable) modes in the dynamics.
{\color{PineGreen}As such, it is important to note that \ac{DGR} is particularly relevant for
online regulation of unstable systems, when the controller does not have
access to enough state data for the purpose of identification or stabilization.}
The proposed technique is close in spirit to modal analysis where regression-based methods are leveraged to extract and control the dominant modes of the system \cite{simon1968theory, proctor2016dynamic}.
The emphasis of \ac{DGR}, however, is on the significance of each temporal action
for safety-critical applications; in these scenarios, it might be rather
unrealistic to generate sufficient data from the inherent unstable modes.

From an implementation perspective, the \ac{DGR} algorithm can become 
computationally expensive for large-scale systems.
This is primary due to steps 7-9 of \Cref{alg:Controller}, where the entire temporal data is stored in $\mathcal{X}_{t+1}$ and $\mathcal{Y}_{t+1}$;
the pseudoinverse operation in the meantime has 
complexity $\mathcal{O}(n^2t)$ required at iteration $t$.
{\color{PineGreen} While for the purpose of analysis, we present the basic form of \ac{DGR} (as in \Cref{alg:Controller}), in \Cref{sec:complexity} we will propose \ac{F-DGR} to circumvent the complexity of storing and computing on large datasets using a rank-one update on the data matrices, resulting in a recursive evaluation of $\mathcal{Y}_{t+1} \mathcal{X}_{t}^{\dagger}$ (see \Cref{alg:Controller_fast}). 
}
%
\subsection{Analysis of DGR}
\label{sec:results}
    
In this subsection, we provide the performance analysis for
\ac{DGR} in the general setting;
as pointed out previously,
\ac{DGR} is particularly relevant when $t\leq n$,
where $n$ denotes the dimension of the underlying system.
We examine the effects of \ac{DGR} on the system's state trajectory and deduce effective guarantees in terms of norm upper-bound and informativity of generated data.
%
In addition, we will see how a particular structure of the system
matrix $A$, such as $\mathcal{R}(A)\subset\mathcal{R}(B)$ or its diagonalizability, facilitates
further insights into the operation of \ac{DGR} as presented in the next subsection.
	%



%
    First, we show why regularizability is essential for the analysis of the trajectory generated under \Cref{alg:Controller}; in hindsight, justifying its introduction in the first place.

{
\begin{lemma}
	\label[lemma]{lem:algo-dyn}
	For all $t>0$, the trajectory generated by
	\Cref{alg:Controller} satisfies,
	\begin{equation*}
	\x_{t+1} = ~\Pi_{\mathcal{R}(B)^\perp} A \x_t + ~\Pi_{\mathcal{R}(B)} \; A \z_t + \Delta_\alpha \w_t ,
	\end{equation*}
	where $\Delta_\alpha \coloneqq B(B^\dagger - G_\alpha) A$, $\z_0 \coloneqq \x_0$, $\w_0 = 0$, and $\z_t \coloneqq ~\Pi_{\mathcal{R}(\mathcal{X}_{t-1})^\perp} \x_t~$ and $\w_t \coloneqq ~\Pi_{\mathcal{R}(\mathcal{X}_{t-1})} \x_t~$ for $t>0$. Furthermore, $\{\z_0,\z_1,\cdots,\z_t\}$ is a set of ``orthogonal'' vectors {(possibly including the zero vector)}, and $\Delta_0 = 0$.
\end{lemma}

	\begin{proof}
		Let $B = U_r \Sigma_r V_r^\intercal$ be the ``thin" SVD of $B$, where $r = \mathrm{rank}(B)$.
		Since $B B^\dagger= U_r U_r^\intercal = \Pi_{\mathcal{R}(U_r)}$, 
		\begin{align*}
    		\x_{t+1} &= A \x_t + B \u_t\\
    		&=\big[A - B G_\alpha \mathcal{Y}_t \mathcal{X}_{t-1}^\dagger \big] \x_t\\
    		&= \left[A - ~B G_\alpha \; A \; ~\Pi_{\mathcal{R}(\mathcal{X}_{t-1})} \right] \x_t\\
    		&= ~\Pi_{\mathcal{R}(U_r)^\perp} A \x_t + B B^\dagger A (\z_t+\w_t) - B G_\alpha \; A \; \w_t\\
    		&= ~\Pi_{\mathcal{R}(U_r)^\perp} A \x_t + ~\Pi_{\mathcal{R}(U_r)} \; A \z_t + B(B^\dagger - G_\alpha) A \w_t.
		\end{align*}
		%
		Thus, the first claim follows as $\mathcal{R}(U_r) = \mathcal{R}(B)$.
		For the second claim, note that the definition of $\z_t$ implies that $\z_t \perp \mathcal{R}(\mathcal{X}_{t-1})$ for all $t>0$, and $\z_k \in \mathcal{R}(\mathcal{X}_{t-1})$ for all $k=1,\ldots,t-1$ and all $t>0$.
		Hence $\{\z_0,\z_1,\cdots,\z_t\}$ consists of orthogonal vectors. Finally, $\Delta_0 = 0$ follows by the definition of $G_\alpha$ and properties of pseudoinverse.
	\end{proof}
}


	
	{The preceding lemma implies that in the case of $\alpha = 0$}, the time series generated by \Cref{alg:Controller} can be considered as the trajectory of a linear system with parameters $(\widetilde{A},\widetilde{B})$ and ``input'' $\z_t$ where,  
	\begin{equation}
    	\label{eqn:DGR-tilde-sys}
    	\widetilde{A} \coloneqq ~\Pi_{\mathcal{R}(B)^\perp} A, \qquad \widetilde{B} \coloneqq ~\Pi_{\mathcal{R}(B)} A,
	\end{equation}
	and $\z_t = \widetilde{K}_t \x_t$, with the {time-varying}, state-dependent feedback gain $\widetilde{K}_t = ~\Pi_{\mathcal{R}(\mathcal{X}_{t-1})^\perp}$.
	%
	{
	Note that, in this case, if $~\Pi_{\mathcal{R}(B)}$ and $A$ commute,\footnote{This is the case if (and only if) both matrices are simultaneously diagonalizable (Theorem 1.3.21 in \cite{horn2012matrix}). If $A$ is symmetric, then these matrices commute if (and only if) they are congruent (Theorem 4.5.15 in \cite{horn2012matrix}).} then $\widetilde{A}\widetilde{B} = 0$ and $\x_{t+1} = \widetilde{A}^{t+1} \x_0 + \widetilde{B} \z_t$.
	%
	Moreover, $\widetilde{A} = 0$ whenever $\mathcal{R}(A) \subset \mathcal{R}(B)$, i.e., the system dynamics will only be driven by the feedback signal $\z_t$; these cases
	will be examined further subsequently.
	In case of general $\alpha$, the system trajectories evolve as,
	\begin{equation}
    	\label{eqn:sys-tilde}
    	\textstyle
    	\x_{t+1} = \widetilde{A}^{t+1} \x_0 + \sum_{r=0}^{t} \widetilde{A}^{t-r} \big[\widetilde{B} \z_r + \Delta_\alpha \w_r\big].
	\end{equation}%
    }%
	Finally, an attractive feature of \ac{DGR} hinges upon the orthogonality
	of the ``hidden" states $\z_t$ generated during the process.

    \subsubsection{Bounding the State Trajectories Generated by DGR}
	\label{subsec:bounds}
	
	\noindent In the open loop setting, the generated data from an unstable system can grow exponentially fast with a rate dictated by the largest unstable mode.
	We show that DGR can prevent this undesirable phenomenon for unstable systems when the system is regularizable.
	The key property for such an analysis involves the notion of instability number.
	
	\begin{definition}
		\label[definition]{def:instNo}
		Given the matrix $A \in \mathbb{R}^{n \times n}$, {\color{PineGreen} for any positive integer $t \leq n$}, its \textit{instability number of order $t$} is defined as,
		\begin{equation*}
		\textstyle
    		M_t(A) \coloneqq \sup_{\{\v_1, \ldots, \v_t\}\in \mathcal{O}_t^n} \hspace{2mm} \|A\v_1\| \, \|A\v_2\| \, \hspace{1mm}\cdots \hspace{1mm}  \|A\v_t\|,
		\end{equation*}
		where $\mathcal{O}_t^n$ is the collection of all sets of $t$ ``orthonormal'' vectors in $\mathbb{R}^n$; {\color{PineGreen} for $t > n$ we define $M_t(A) = 0$.}
	\end{definition}
	
	\begin{figure}
		\centering
		\includegraphics[width = 0.4\textwidth]{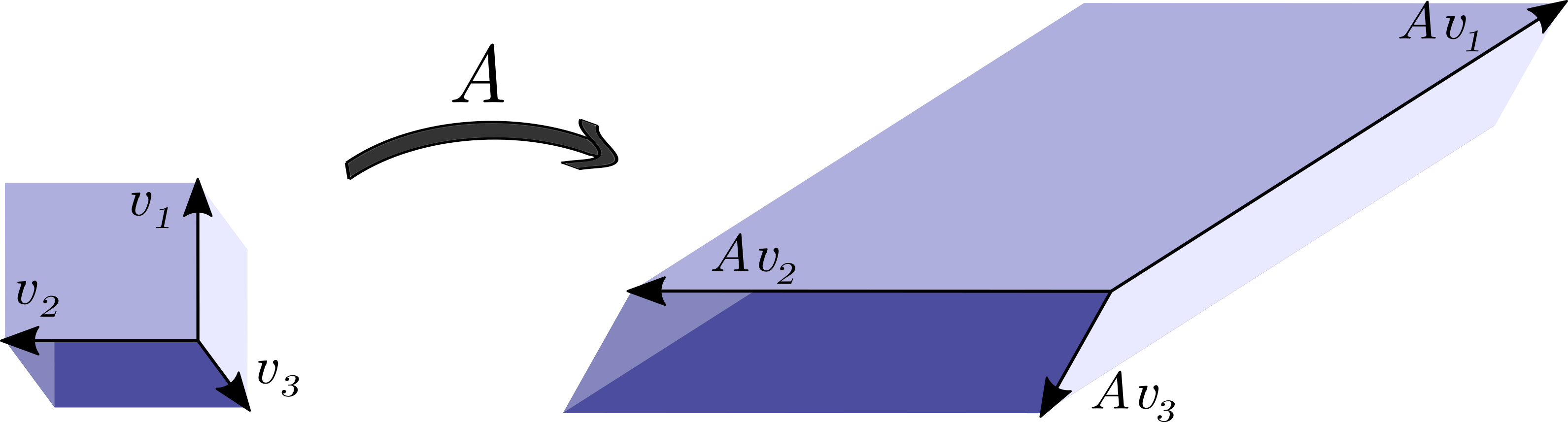}
		\caption{A unit cube in the domain of $A$ that is mapped to a parallelepiped in its range space.}
		\label{fig:map}
		\vspace{-0.3cm}
	\end{figure}
	
	{\color{Blue} Note that $M_t(A)\leq \|A\|^t$ for all $t$, where $\|.\|$ denotes the induced operator norm. However the behavior of $M_t(A)$ is fundamentally distinct from $\|A\|^t$.}
	In fact, the instability number of a matrix is distinct 
	from products of any subset of its eigenvalues.
	Consider for example, a $t$-dimensional hypercube with its image under $A$ as a parallelotope (see \Cref{fig:map} for a 3D schematic).
	The instability number is related to the multiplication of the lengths of edges radiating from one vertex of the parallelotope, while $\det (A^\intercal A)$ is related to its volume.
	The instability number of a matrix can in fact be difficult to compute.
	In what follows, we first provide upper and lower bounds on $M_t(A)$ characterizing its growth rate with respect to the largest singular value of $A$.
	Subsequently, these bounds will be used to provide a bound on the norm of the state trajectory generated by \ac{DGR}.
	
	\begin{lemma}
	    \label[lemma]{lem:Mbound-gen}
		Let $\sigma_1, \cdots, \sigma_n$ denote the singular values of $A\in\mathbb{R}^{n\times n}$ in a descending order.
		Then for $t \leq n$,
		\begin{gather*}
    		\left[ \frac{\sigma_1^2}{t} \right]^t \leq M_t^2(A) \leq \left[ \frac{\sigma_1^2}{t} \right]^t +  \sum_{j=1}^{t-1} \left[\frac{\sigma_1^2}{t-j} \right]^{t-j} \binom t j \delta^j + \delta^t,
		\end{gather*}
		where $\delta \coloneqq \sum_{i=2}^t \sigma_i^2$, with $M_t(A)$ as defined in \Cref{def:instNo}.
	\end{lemma}
	
	
    The lower and upper bounds in \Cref{lem:Mbound-gen} show that, particularly when $\delta<1$, $M_t(A)$ initially grows similar to $( {\sigma_1}/{\sqrt{t}} )^t$ for $t\leq \sigma_1^2$, in contrast to the exponential growth of $\|A\|^t = \sigma_1^t$. 
	{\color{Blue}
	This difference becomes more pronounced for $t>\sigma_1^2$ when $( {\sigma_1}/{\sqrt{t}} )^t$ starts \emph{decreasing}.
    This fact is illustrated via an example in \Cref{fig:Mt}, where the first five dominant terms of the upper bound are plotted and the green region shows where the actual value of $M_t^2(A)$ lies.
    \begin{figure}[pt]
    \color{Blue}
        \centering
        \includegraphics[width=\columnwidth]{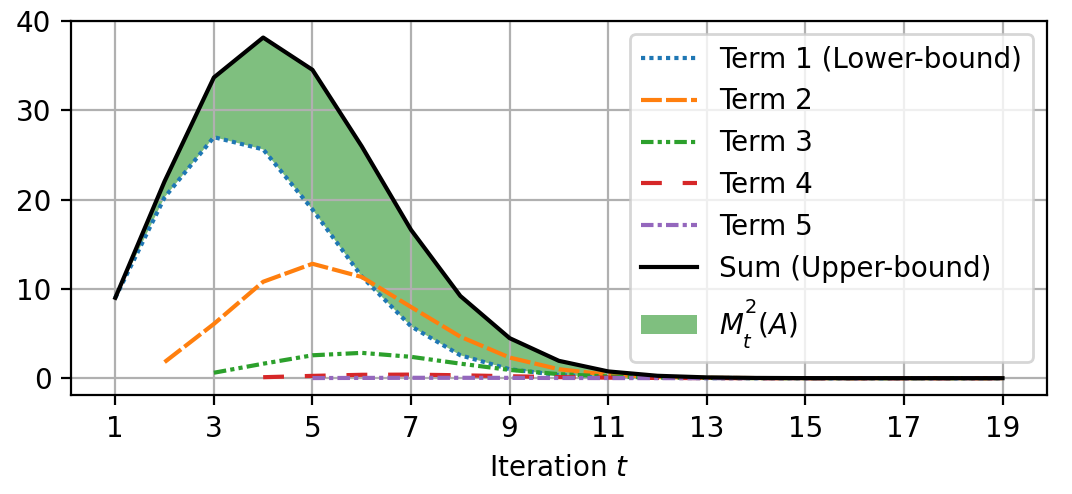}
        \caption{Illustration of the upper and lower bounds for instability number of a system with $\sigma_1 = 3$ and $\delta = 0.1$ as in \Cref{lem:Mbound-gen}.}
        \label{fig:Mt}
        \vspace{-0.3cm}
    \end{figure}
	The following result provides an upper bound on the state trajectories for the most general case through the lens of regularizability.
	}

	\begin{theorem}
    \label{thm:recursion}
	 
        For any regularizable pair $(A,B)$, the trajectory generated by \Cref{alg:Controller} satisfies the following bound for all $t>0$,
		\begin{equation*}
		\|\x_{t+1}\| \leq L_{t+1} \|\x_0\| \,,
		\end{equation*}
		where $L_t$ satisfies the recursion,
		\[ L_{t+1} = a_t + \sum_{r=1}^{t} b_{t,r} L_r, \hspace{5mm} L_1 = \|A \overline{\z}_0\|, \]
		with $$b_{t,r} =\sqrt{ \|\widetilde{A}^{t-r}\widetilde{B}\overline{\z}_r\|^2 + \|\widetilde{A}^{t-r}\Delta_\alpha \overline{\w}_r\|^2},$$ and $a_t = \|\widetilde{A}^t A \overline{\z}_0\|$, 
		where $\overline{\z}_r = \z_r / \|\z_r\|$ (if $\z_r \neq 0$, otherwise $\overline{\z}_r = 0$), and $\overline{\w}_r$ is similarly defined.
	\end{theorem}{}
    
	\begin{proof}
		Knowing that $\x_1= A \x_0$, it follows that $\|\x_1\| \leq L_1 \|\x_0\|$.
		Furthermore, for $t\geq 1$, \cref{eqn:sys-tilde} leads to,
		\begin{equation*}
    		\textstyle \x_{t+1}=\widetilde{A}^{t} A \x_0 + \sum_{r=1}^{t}
    		\widetilde{A}^{t-r} \left[\widetilde{B} \z_r + \Delta_\alpha \w_r\right],
		\end{equation*}
		since $\widetilde{A} + \widetilde{B} = A$ and $\w_0=0$ by definition.
		This implies that,
		\begin{align*}
    		\|\x_{t+1}\| 
    		\leq & \, \|\widetilde{A}^{t} A \x_0\| \\
    		&\textstyle + \sum_{r=1}^{t} \|\widetilde{A}^{t-r} \widetilde{B} \overline{\z}_r\|\|\z_r\| + \|\widetilde{A}^{t-r} \Delta_\alpha \overline{\w}_r\|\|\w_r\|,\\
    		\leq & \textstyle  \, a_t \|\x_0\| + \sum_{r=1}^{t} b_{t,r} \|\x_r\|,
		\end{align*}
		where we have used Cauchy–Schwarz inequality in conjunction with the equality $\|\z_r\|^2 + \|\w_r\|^2 = \|\x_r\|^2$.
		Using this recursive bound, the rest of the proof follows by induction.
	\end{proof}
	\begin{remark}
	    \label{rmk:bound_behavior}
	    Note that in the analysis above, when the system is regularizable, $a_t$ eventually decreases exponentially fast as $t$ increases.
	    Furthermore, the term $b_{t,r}$ in the sum increases as $r$ approaches a fixed $t$.
	    %
		%
		Finally, one can show that the obtained upper bound is tight by considering \Cref{ex:motivateExample} with $\lambda_1>0$ and $\lambda_i = 0$ for $i>1$.
	\end{remark}{}
    {\color{Blue}
	Note that computing/estimating the upper bound in \Cref{thm:recursion} requires knowledge on the matrix $A$, making these estimates more practical for structured systems (for example, see \Cref{cor:upperbound} and \Cref{rem:simple-upperbound}).
    }
	In order to shed light on the intuition behind this upper bound, we next study simpler cases {with $\alpha = 0$,} where there exists small enough $\kappa$ for which
	$b_{t,r} \leq \|\widetilde{A}^{t-r}\widetilde{B}\| \leq \kappa$ for all $r<t$.
	In particular, we can show that if the system is regularizable and $\widetilde{A}\widetilde{B}= 0$, then the trajectories of the closed loop system
	will be bounded by a combination of instability number of different orders. 
	This is stated in the following corollary of \Cref{thm:recursion}.
	
	\begin{corollary}
	    \label[corollary]{cor:upper-bound-gen}
		For any  regularizable pair $(A,B)$ with $\widetilde{A} \widetilde{B} = 0$, and $M_t(A)$ as in \Cref{def:instNo}, {the system trajectory generated by \Cref{alg:Controller} with $\alpha = 0$} satisfies the following for all $t>0$,
		\begin{equation*}
		\frac{\|\x_{t+1}\|}{\|\x_0\|} \leq a_t + \sum_{r=1}^{t-1} M_r(A) a_{t-r} + M_{t+1}(A)\,.
		\end{equation*}
	\end{corollary}
	\begin{proof}
	    \color{PineGreen}
		For brevity, let $b_t = b_{t,t}$, then as $\widetilde{A} \widetilde{B} = 0$, the recursion in \Cref{thm:recursion} reduces to $L_{t+1} = a_t + b_t L_t$ with $L_1 = \|A \overline{\z}_0\|$; and its solution has the following form for all $t> 0$,
		\begin{equation}\label{eq:L-spec-proof}
		\textstyle
		    L_{t+1} =a_t+ b_t \cdots b_2 b_1 L_1 + \sum_{r=1}^{t-1} b_t \, \cdots \, b_{t+1-r} \, a_{t-r}.
		\end{equation}
		As $\alpha = 0$ and orthogonal projection is non-expansive, we claim that $b_r = \|\widetilde{B} \overline{\z}_r\| \leq \|A \overline{\z}_r\|$ which vanishes whenever $\z_r = 0$.
        In the meantime, by \Cref{lem:algo-dyn}, $\{\z_r\}_0^t$ must be a set of orthogonal vectors for any $t > 0$, and thus $\{\overline{\z}_r\}_0^t$ is a set of orthogonal vectors that are either normal or zero. 
        Note that if $t \geq n$, then $\{\overline{\z}_r\}_0^t$ must contain at least one zero vector for dimensional reasons.
        Therefore, by \Cref{def:instNo}, we conclude that
		\( b_t \, \cdots \, b_{t+1-r} \leq 
		M_{r}(A),\)
		for each $r= 1,\dots, t-1$.
		Similarly, as $L_1 = \|A \overline{\z}_0\|$, we have
		\(b_t \, \cdots \, b_{2} b_{1} L_1 \leq
		M_{t+1}(A).\)
		By using these inequalities in \cref{eq:L-spec-proof},  the claim follows by \Cref{thm:recursion}.
		\vspace{-0.2cm}
	\end{proof}
    The above observation further highlights the importance of the instability number in the context of \ac{DGR}. {\color{PineGreen} Note that the terms in the upper bound involving $M_r(A)$ vanishes if $r > n$.}
\subsubsection{Informativity of the \ac{DGR} Generated Data}
\label{subsec:informativity}

{\color{Blue} In the sequel, we show that \ac{DGR} generates linearly independent state-trajectory data; we refer to this as informativity of data.}
We then proceed to make a connection between this independence structure 
and the number of excited modes in the system.
Before we proceed, let us define $L_k^t(A)$, that is based on eigenvalues corresponding to $k$ modes of a matrix $A$, as,
\begin{equation}
\label{eqn:LA}
L_k^t(A) \coloneqq \scalemath{0.9}{\left(
\begin{array}{*4{c}}
1& \lambda_{1} & \cdots & (\lambda_{1})^{t-1} \\
1& \lambda_{2} & \cdots & (\lambda_{2})^{t-1} \\
\vdots& \vdots & \ddots & \vdots  \\
1& \lambda_{k} & \cdots & (\lambda_{k})^{t-1} \\
\end{array}\right)}, \quad 1\leq t\leq n.
\end{equation}
\begin{remark}
Note that $L_k^t(A)$ has a specific structure that hints
at its invertibility.
In fact, for $t=k$, $L_k^k(A)$ is the Vandermonde matrix formed by $k$ eigenvalues of $A$ which would be invertible if and only if $\lambda_1, \cdots, \lambda_k$ are distinct.
More generally, if $\{ \lambda_1, \cdots, \lambda_k \}$ consists of $r$ distinct eigenvalues (where $r\leq k$), then $L_k^r(A)$ has full column rank. 
\end{remark}
{\color{Blue}
Intuitively, informative data--due to its linear independence structure--contain useful information for decision-making purposes.}
{\color{violet} In particular, if the choice of $\x_0$ results in exciting all modes of $A$, one might expect that a useful online regulation algorithm should generate informative data at the same time that it is regulating the state-trajectory. 
The next theorem formalizes how \ac{DGR} realizes this expectation depending on what modes of the system are excited by the initial condition.}


\begin{theorem}
    \label{thm:linIndp}
    Let $\x_0$ excite $k_1+k_2$ modes of $A$, such that $k_1$ modes are in $\mathcal{R}(B)$ and $k_2$ modes are in $\mathcal{R}(B)^\perp$.
    If the excited modes correspond to distinct eigenvalues, then  $\{\x_0,\dots,\x_{r-1}\}$, generated by \Cref{alg:Controller} with $\alpha = 0$, is a set of linearly independent vectors for any $r \leq \max\{k_1,k_2\}$.
\end{theorem}

\begin{proof}
    Without loss of generality, let $\lambda_1,\dots, \lambda_{k_1}$ be the eigenvalues corresponding to the excited modes $u_1, \dots, u_{k_1} \in \mathcal{R}(B)$, and similarly $\lambda_{k_1+1},\dots, \lambda_{k_1+k_2}$ be corresponding to $u_{k_1+1}, \dots, u_{k_1+k_2} \in \mathcal{R}(B)^\perp$.
    Recall that $\mathcal{X}_{t-1} = [\x_0\ \x_1\ \dots\ \x_{t-1}]$; then by definition of $\z_t$ in \Cref{lem:algo-dyn}, for $t\geq 1$ there exists scalar coefficients $ \zeta^t_{0}, \cdots, \zeta^t_{t-1} \in \mathbb{R}$ such that $\z_t = \x_t - \sum_{j=0}^{t-1} \zeta_j^t \x_j$.  This together with the dynamics in \Cref{lem:algo-dyn} imply that $\x_1 = A \x_0$ and for $t\geq 2$,
    \begin{equation}\label{eqn:xt-proof}
         \textstyle \x_{t}=  A \x_{t-1}  -  \Pi_{\mathcal{R}(B)} \sum_{j=0}^{t-2} \zeta_j^t A \x_j.
    \end{equation}
    Since $\x_0$ excites $k_1+k_2$ modes of the system, we have $\x_0 = \sum_{\ell=1}^{k_1+k_2} \beta_\ell \u_\ell$, where $\beta_\ell$ are some nonzero real coefficients and $(\lambda_\ell,\u_\ell)$ are eigenpairs of $A$.
    Hence $\x_1 = A \x_0 = \sum_{\ell=1}^{k_1+k_2} \beta_\ell \lambda_\ell \u_\ell$, and we claim that for $t \geq 2$ there exist scalar coefficients $ \xi^t_{1}, \cdots, \xi^t_{t-1} \in \mathbb{R}$ such that,
    \begin{gather}    \label{eqn:xspan-gen}
        \x_t = \hspace{-0.1cm} \sum_{\ell=1}^{k_1} \beta_\ell \Big[ (\lambda_\ell)^t  -  \hspace{-0.1cm} \sum_{i=1}^{t-1} \xi^{t}_{i}  (\lambda_\ell)^i   \Big]\u_\ell + \hspace{-0.3cm} \sum_{\ell=k_1+1}^{k_1+k_2} \beta_\ell (\lambda_\ell)^t \u_\ell.
    \end{gather}
    The proof of the last claim is by induction.
    Note that $A^t \x_0 = \sum_{\ell=1}^{k_1+k_2} \beta_\ell (\lambda_\ell)^t \u_\ell$, and by substituting this into \cref{eqn:xt-proof} for $t=2$ we have that,
    \begin{align*}
        \x_2 
        &= A \x_1 -\zeta_0^2 \Pi_{\mathcal{R}(B)} A\x_0 \\
        &=  \Pi_{\mathcal{R}(B)} \left[A^2 \x_0 - \zeta_0^2 A \x_0\right] +  \Pi_{\mathcal{R}(B)^\perp} A^2 \x_0\\
        &= \textstyle  \sum_{\ell=1}^{k_1} \beta_\ell  \left[(\lambda_\ell)^2 - \zeta_0^2 \lambda_\ell \right] \u_\ell +  \sum_{\ell=k_1+1}^{k_1+k_2} \beta_\ell (\lambda_\ell)^2 \u_\ell,
    \end{align*}
    where the last equality is due to the fact that $\u_\ell \in \mathcal{R}(B)$ for $\ell \leq k_1$ and $\u_\ell \in \mathcal{R}(B)^\perp$ for $\ell > k_1$.
    By choosing $\xi_1^2 = \zeta_0^2$, we have shown that \cref{eqn:xspan-gen} holds for $t=2$.
    Now suppose that \cref{eqn:xspan-gen} holds for all $2,\dots, t-1$; it now suffices to show that this relation also holds for $t$. By substituting the hypothesis for $2,\dots, t-1$ into \cref{eqn:xt-proof},
     \begin{align*}
         \x_t 
         =& \textstyle \sum_{\ell=1}^{k_1} \beta_\ell \Big[ (\lambda_\ell)^{t}  -  \sum_{i=1}^{t-2} \xi^{t-1}_{i}  (\lambda_\ell)^{i+1}   \Big]\u_\ell \\
         & \textstyle + \sum_{\ell=k_1+1}^{k_1+k_2} \beta_\ell (\lambda_\ell)^{t} \u_\ell - \sum_{\ell=1}^{k_1} \beta_\ell[ \zeta_0^t \lambda_\ell + \zeta_1^t(\lambda_\ell)^2] \u_\ell \\
         & \textstyle -\sum_{j=2}^{t-2} \zeta_j^t \sum_{\ell=1}^{k_1} \beta_\ell \Big[ (\lambda_\ell)^{j+1}  -  \sum_{i=1}^{j-1} \xi^{j}_{i}  (\lambda_\ell)^{i+1}   \Big]\u_\ell.
     \end{align*}
     Therefore, $\x_t 
         = \textstyle \sum_{\ell=1}^{k_1} \beta_\ell \left[ \star \right]\u_\ell +  \sum_{\ell=k_1+1}^{k_1+k_2} \beta_\ell (\lambda_\ell)^{t} \u_\ell,$
    where $\star$ replaces the expression,
    \begin{equation*}
         (\lambda_\ell)^{t}  -  \sum_{i=1}^{t-2} \xi^{t-1}_{i}  (\lambda_\ell)^{i+1} - \sum_{j=0}^{t-2} \zeta_j^t(\lambda_\ell)^{j+1}+ \sum_{j=2}^{t-2} \sum_{i=1}^{j-1} \zeta_j^t \xi^{j}_{i}  (\lambda_\ell)^{i+1}.
     \end{equation*}
    By appropriate choices of $ \xi^t_{1}, \cdots, \xi^t_{t-1} \in \mathbb{R}$, we can rewrite $\star = (\lambda_\ell)^{t} -  \sum_{i=1}^{t-1} \xi^{t}_{i}  (\lambda_\ell)^i $.
    This completes the proof of the claim in \cref{eqn:xspan-gen} by induction. 
    Now, let $\widehat{\x} = \sum_{j= 0}^{r-1} \gamma_j \x_j$ for some $\gamma_j \in \mathbb{C}$ and some $r\leq \max\{k_1,k_2\}$.
    Then, by substituting $\x_j$ from (\ref{eqn:xspan-gen}) and exchanging the sums over $j$ and $\ell$ we have, 
    \begin{gather*}
    	\textstyle \widehat{\x} =
    	\sum_{\ell=1}^{k_1}  \beta_\ell \Big[\gamma_0  + \gamma_1 \lambda_\ell + \sum_{j=2}^{r-1} \gamma_j\big[ (\lambda_\ell)^j   -  \sum_{i=1}^{j-1} \xi^{j}_{i}  (\lambda_\ell)^i  \big] \Big] \u_\ell \\
    	\textstyle + \sum_{\ell=k_1+1}^{k_1+k_2} \beta_\ell \sum_{j=0}^{r-1} \gamma_j (\lambda_\ell)^j \u_\ell .
    \end{gather*}
    Now, by exchanging the sums over $i$ and $j$, it follows that,
    \begin{gather*}
    \textstyle	\widehat{\x} = \sum_{\ell=1}^{k_1} \beta_\ell  \Big[ \gamma_0 + \sum\limits_{i=1}^{r-2} \big[ \gamma_i - \sum\limits_{j=i+1}^{r-1} \gamma_j \xi_{i}^{j} \big] (\lambda_{\ell})^i + \gamma_{r-1} (\lambda_{\ell})^{r-1} \Big] \u_\ell  \\
     \textstyle   +\sum_{\ell=k_1+1}^{k_1+k_2} \beta_\ell \big[ \sum_{j=0}^{r-1} \gamma_j (\lambda_\ell)^j \big] \u_\ell .
    \end{gather*}
    Since $\{\u_\ell\}_1^{k_1+k_2}$ are eigenvectors associated with distinct eigenvalues, they are linearly independent.
    Thus, noting that $\beta_\ell \neq 0$ for all $\ell = 1, \cdots, k_1+k_2$, then $\widehat{\x}= 0$ implies that,
    \begin{equation*}
    \textstyle
        \gamma_0 + \sum\limits_{i=1}^{r-2} \Big[ \gamma_i - \sum\limits_{j=i+1}^{r-1} \gamma_j \xi_{i}^{j} \Big] (\lambda_{\ell})^i + \gamma_{r-1} (\lambda_{\ell})^{r-1} = 0,
    \end{equation*}
for all $\ell = 1, \dots,k_1$; and $\sum_{j=0}^{r-1} \gamma_j (\lambda_\ell)^j=0$,
    for all $\ell = k_1+1, \dots,k_1+k_2$. By rewriting the last two sets of equations in matrix form we get,
    \begin{equation}\label{eqn:gamma-zero}
           \left(\begin{array}{*2{c}}
    		L_{k_1}^r(A)(I-\Xi)  \\
            \widehat{L}_{k_2}^r(A)
        	\end{array}\right) \ogamma= 0,
    \end{equation}
    where $\widehat{L}_{k_2}^r(A)$ is the last $k_2$ rows of $L_{k_1+k_2}^r (A)$ and
    \small
    \begin{equation} \label{eq:xi-gamma-def}
     \Xi \coloneqq \left(\begin{array}{*6c}
    		0 & 0 & 0 & 0 & \dots & 0\\
            0 & 0 & \xi_1^2 & \xi_1^3 & \cdots & \xi_1^{r-1}\\
            0 & 0 & 0 &   \xi_2^3 & \cdots & \xi_2^{r-1}       \\
             0 & 0 & 0 &0  & \ddots & \vdots      \\
          \vdots & \vdots & \vdots & \vdots & \ddots & \xi_{r-1}^{r-1}  \\
            0&0&0&0&\cdots&0
    		\end{array}\right), \; \ogamma \coloneqq \left(\begin{array}{*4c}
    		\gamma_0\\
            \gamma_1\\
            \\
            \vdots\\
            \\
            \gamma_{r-1}
    		\end{array}\right).
    \end{equation}
    \normalsize
    Note that $I - \Xi$ is invertible by construction.
    Since the excited modes correspond to distinct eigenvalues, if $r\leq \max\{k_1,k_2\}$, then either 
    $L_{k_1}^{r}(A)$ or $\widehat{L}_{k_2}^r(A)$ has full column rank.
    Either way, \cref{eqn:gamma-zero} implies that $\ogamma =0$ and thus $\{\x_0, \dots, \x_{r-1}\}$ is a set of linearly independent vectors.
    This observation completes the proof as $r\leq \max\{k_1,k_2\}$ was chosen arbitrary.
\end{proof}
    {\color{violet}
    The preceding theorem guarantees 
    the linear independence of the state-trajectory generated by \ac{DGR} whenever the exited modes lie in $\mathcal{R}(B)$ or $\mathcal{R}(B)^\perp$, even though our observations suggest that it must remain valid for arbitrary excitation of the modes. Nonetheless, \ac{DGR} remains effective in terms of online regulation from an arbitrary choice of $x_0$ as guaranteed in \Cref{thm:recursion}, \Cref{cor:upper-bound-gen}, and subsequently in \Cref{cor:upperbound}.
    }
    \subsection{\color{violet}Special Case of $\mathcal{R}(A) \subset \mathcal{R}(B)$ with $\alpha = 0$}
	\label{sec:results-rank-n-B}
	
    In order to better understand the behavior of \ac{DGR}, in this subsection, we study the more special case where $\mathcal{R}(A) \subset \mathcal{R}(B)$. 
    This includes the case where $\mathrm{rank}(B)=n$, i.e., one 
    can directly control each state of the system (e.g. see \cite{sharf2018network,friedkin1990social}).
	Note that $\mathcal{R}(A) \subset \mathcal{R}(B)$ implies that
	$\widetilde{A}= 0$ which, in turn, results in regularizability of $(A,B)$.
	This, together with \Cref{cor:upper-bound-gen}, results in the following corollary.
	
	\begin{corollary}
	    \label[corollary]{cor:upperbound}
		For any matrix  $A \in \mathbb{R}^{n\times n}$ and $B \in \mathbb{R}^{n \times m}$, where $\mathcal{R}(A) \subseteq \mathcal{R}(B)$, the trajectory generated by 
		\Cref{alg:Controller} {with $\alpha = 0$} satisfies the following for all $t>0$,
		\begin{equation*}
		\|\x_{t+1}\| \leq M_{t+1}(A) \|\x_0\| \, ,
		\end{equation*}
		with $M_t(A)$ as in \Cref{def:instNo}.
	\end{corollary}{}
	
	\begin{proof} 
    	Note that $\widetilde{A} = ~\Pi_{\mathcal{R}(B)^\perp} A = 0$ whenever $\mathcal{R}(A) \subseteq \mathcal{R}(B)$.
    	The claim now follows by \Cref{cor:upper-bound-gen} since $\widetilde{A}\widetilde{B} = 0$ and $a_k = 0$ for all $k = 1, \dots, t$.
	\end{proof}
	
	{\color{PineGreen} Note that under the hypothesis of \Cref{cor:upperbound}, in particular $\x_{t+1} = 0$ for all $t \geq n$ whenever \ac{DGR} is in effect for the noiseless dynamics in \cref{eqn:sys-dynamic} (see \Cref{fig:algo}); however, this might happen even before $t$ reaches $n$ as will be discussed in \Cref{prop:zero}.}
	Also, the latter bound becomes more structured for a symmetric $A$ by combining the results from \Cref{cor:upperbound} and \Cref{lem:Mbound-gen} which we skip for brevity.
	
	\begin{remark}\label{rem:simple-upperbound}
		In order to further illustrate the bound stated in \Cref{cor:upperbound}, assume that $\delta e \leq 1$.
		Then, from \Cref{lem:Mbound-gen},
		\begin{gather*}
    	\frac{\|\x_t\|^2}{\|\x_0\|^2} \leq \left[\frac{\sigma_1^2}{t}\right]^t + \sum\limits_{j=1}^{\lfloor t/2 \rfloor} \left[\frac{\sigma_1^2 }{t-j} \right]^{t-j}\left[\frac{t}{j}\right]^j
    		+ \hspace{-0.3cm} \sum\limits_{j=\lfloor t/2 \rfloor + 1}^{t-1} \left[\frac{t \, \sigma_1^2 }{(t-j)^2} \right]^{t-j} \hspace{-0.3cm} +1,
		\end{gather*}
		\noindent where we have also used $\binom t j \leq (e \, t/j)^j$.
		This implies that as $t$ gets larger than $\sigma_1^2$, the terms with large powers admit smaller bases and those with large bases will gain smaller powers comparing to $\sigma_1^{2t}$.
		This is despite the fact that for small $t$, the relative norm of the state might grow.
	\end{remark}{}
	
	\begin{figure}
		\centering
		\includegraphics[width=0.7\columnwidth]{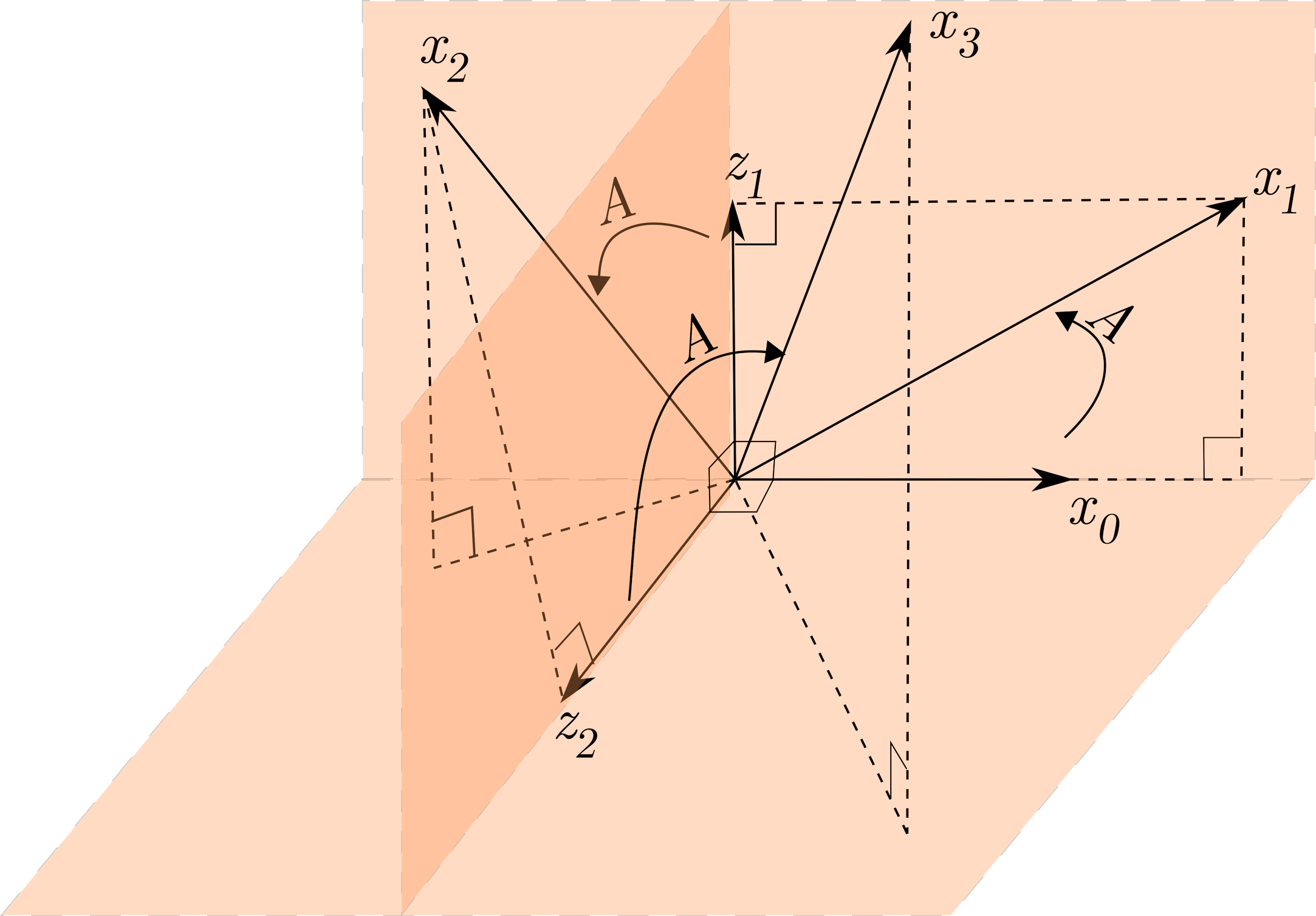}
		\caption{A geometric schematic of \ac{DGR} when $\mathcal{R}(A) \subseteq \mathcal{R}(B)$. Since $\z_0\coloneqq \x_0$, $\z_t\perp \mathcal{R}(\mathcal{X}_{t-1})$ and $\z_{t}\in\mathcal{R}(\mathcal{X}_{t})$ for $t=1,2$, the set $\{\z_0,\z_1,\z_2\}$ consists of orthogonal vectors.}
		\label{fig:algo}
		\vspace{-0.3cm}
	\end{figure}
	
	In the sequel, as a result of linear independence established in \Cref{thm:linIndp} we show how the simplified bounds (derived in \Cref{sec:results}) clarify the elimination of the unstable modes in the system.
	\begin{corollary}
		\label[corollary]{cor:linIndp-special}
		Suppose 
		$\mathcal{R}(A) \subseteq \mathcal{R}(B)$ and let $\x_0$ excite $k$ modes of $A$.
		If $r$ eigenvalues corresponding to the $k$ excited modes are distinct for some $r\leq k$, then $\{\x_0,\dots,\x_{r-1}\}$, generated by \Cref{alg:Controller} with $\alpha = 0$, is a set of linearly independent vectors.
	\end{corollary}
	
	\begin{proof}
		Given that $\mathcal{R}(A) \subseteq \mathcal{R}(B)$, all the modes of $A$ are contained in $\mathcal{R}(B)$, so without loss of generality, let $\lambda_1,\dots, \lambda_{k}$ be the eigenvalues corresponding to the excited modes $u_1, \dots, u_{k} \in \mathcal{R}(B)$.
		Then, following the proof of \Cref{thm:linIndp}, \Cref{eqn:xspan-gen} reduces to,
		\begin{equation*}
        \textstyle
        \x_t =  \sum_{\ell=1}^{k} \beta_\ell \left[ (\lambda_\ell)^t  -  \sum_{j=1}^{t-1} \xi^{t}_{j}  (\lambda_\ell)^j   \right]\u_\ell.
        \end{equation*}
        Now, let $\widehat{\x} = \sum_{j= 0}^{r-1} \gamma_j \x_j$ for some $\gamma_j \in \mathbb{C}$ and $r\leq k$.
        Then following the same argument as in the proof of \Cref{thm:linIndp} about $\widehat{\x}$, \cref{eqn:gamma-zero} reduces to $L_{k}^r(A)(I-\Xi) \ogamma= 0$,
        with similar definitions of $\Xi$ and $\ogamma$ as in \cref{eqn:LA}, and $L_{k}^r(A)$ as in \cref{eq:xi-gamma-def}.
        Since $r$ eigenvalues corresponding to $k$ excited modes are distinct, $L_k^{r}(A)$ has full column rank.
        As $I-\Xi$ is invertible, we conclude that $\ogamma = 0$ meaning that $\{\x_0, \dots, \x_{r-1}\}$ are linearly independent.
	\end{proof}

	An immediate consequence of the above corollary is that \ac{DGR} generates data that is effective for simultaneous identification of modes even with multiplicity greater than one.
	 
	\begin{proposition}
	    \label[proposition]{prop:zero}
		Suppose that $A$ is diagonalizable with $\mathcal{R}(A) \subseteq \mathcal{R}(B)$, and
		let $\x_0$ excite $k$ modes of $A$ corresponding to $r$ distinct eigenvalues (where possibly, $r \leq k$).
		Then, in exactly $r$ iterations {of \Cref{alg:Controller} with $\alpha = 0$}, $\mathrm{span}\{\x_0,\dots,\x_{r-1}\}$ coincides with the subspace containing these excited modes; furthermore, $\x_{r +1}=0$.
	\end{proposition}
	
	\begin{proof}
		 Without loss of generality, let $\x_0$ excite the $k$ modes of $A$ corresponding to $\lambda_1, \cdots, \lambda_{r}$.
		Since $A$ is diagonalizable, let $A = U \Lambda U^{-1}$ be its eigen-decomposition and so $\x_0$ excite $\{\u_1,\cdots, \u_k\}$, i.e., $\x_0=\sum_{i=1}^k \beta_i \u_i$, with $\beta_i\not=0$.
		Define
		\begin{align*}
		    \mathcal{I}(\lambda_i) = \left\{j: \u_j~\text{is the eigenvector corresponding to}~\lambda_i\right\},
		\end{align*}
		for $i=1,\cdots, r$.
		Furthermore, define the $r$-dimensional subspace,
		\begin{align*}
		\textstyle
		    S \coloneqq \mathrm{span}\left\{\sum_{j\in \mathcal{I}(\lambda_1)}\beta_j \u_j,\ \cdots, \sum_{j\in \mathcal{I}(\lambda_{r})} \beta_j \u_j \right\},
		\end{align*}
		where the span is taken over the complex field.
		We prove by induction that $\x_t \in S$ for all $t = 1,\cdots,r$.
		Notice that $\x_0 \in S$ and suppose that $\{\x_0,\dots,\x_{t-1}\} \subset S$; recall from the proof of \Cref{cor:linIndp-special} that $\x_t = A\z_{t-1}$, where $\z_{t-1} =  ~\Pi_{\mathcal{R}(\mathcal{X}_{t-2})^\perp} (\x_{t-1})$.
		Since $\x_{t-1}\in S$ and $\mathrm{span}\{\x_0,\dots,\x_{t-2}\} \subset S$, one can conclude that $\z_{t-1}\in S$, and from the definition of $S$, $\x_{t}=A\z_{t-1}\in S$.
		\noindent On the other hand, 
		since $\lambda_1, \lambda_2, \cdots, \lambda_r$ are distinct eigenvalues, by \Cref{cor:linIndp-special}, $\textbf{dim}\left(\mathrm{span}\{\x_0,\dots,\x_{r-1}\}\right) = r$.
		By hypothesis of the induction $\mathrm{span}\{\x_0,\dots,\x_{r-1}\} \subset S$, and since $\textbf{dim}(S)=r$, we conclude that $\mathrm{span}\{\x_0,\dots,\x_{r-1}\}$ must be the entire $S$, i.e. $\mathrm{span}\{\x_0,\dots,\x_{r-1}\} = S$, proving the first claim.
		Lastly, $\z_r =~\Pi_{\mathcal{R}(\mathcal{X}_{r-1})^\perp}(\x_r)=0$ since $\x_r \in S$, and thus $\x_{r+1} = A \z_r = 0$, thereby completing the proof.
		\end{proof}
    
    Following \Cref{prop:zero}, if $\x_0$ excites $k$ modes of the system corresponding to distinct eigenvalues with trivial algebraic multiplicities, then \Cref{alg:Controller} identifies all the excited modes of the system in $k$ iterations. Furthermore, this implies that $\x_{k +1}=0$, 
    i.e., \ac{DGR} eliminates the unstable modes and regulates the unknown system in exactly $k+1$ iterations.
    {\color{PineGreen} This also implies that, if $k< n$ then online regulation of the system is achieved, even before enough data is available for full identification of system parameters.}
    

\subsection{\color{PineGreen} Boosting the Performance of \ac{DGR}}
\label{sec:complexity}

\ac{DGR} as introduced in \Cref{alg:Controller} can become computationally burdensome for large-scale systems.
This is mainly due to storing the entire history of data in $\mathcal{X}_t$ and $\mathcal{Y}_t$ followed by the update of the controller that finds the pseudoinverse as well as multiplication of these data matrices (steps 7-9).
Assuming the \ac{SVD}-based computation of pseudoinverse, the complexity of the method is\footnote{The multiplication $\mathcal{Y}_{t+1} \mathcal{X}_t^\dagger$ requires another $\mathcal{O}(n^2t)$ operations that can be significant for large $n$.} $\mathcal{O}(n^2 t)$.
In this section, we show that such computational burden can be circumvented using rank-1 modifications of data matrices as a result of the discrete nature of data collection in our setup.
Note that for computing $K_{t+1}$  from \eqref{eqn:opt-input}
we only need to access
$\mathcal{Y}_{t+1} \mathcal{X}_t^\dagger$
(rather than $\mathcal{X}_t^\dagger$).
To this end, we leverage the results of \cite{Meyer1973generalized} in order to find $\mathcal{Y}_{t+1} \mathcal{X}_t^\dagger$ recursively as a function of $\mathcal{Y}_{t} \mathcal{X}_{t-1}^\dagger$, $\mathcal{X}_{t-1} \mathcal{X}_{t-1}^\dagger$, and $\x_t$.

\begin{proposition}
    \label[proposition]{prop:pseudo}
    
    Let $\mathcal{X}_{t-1}$ be as in \Cref{alg:Controller}, $\x_t$ be the state measurement at iteration $t$ and $\z_t = ~\Pi_{\mathcal{R}(\mathcal{X}_{t-1})^\perp} \x_t~$.
    If $\x_t\not\in\mathcal{R}(\mathcal{X}_{t-1})$ then
    \begin{equation}
        \label{eq:pseudo_informative}
		\scalemath{0.9}{\mathcal{X}_t^\dagger = \left(\begin{array}{*1{c}}
			\mathcal{X}_{t-1}^\dagger - \ogamma_t \z_t^\dagger \\ \z_t^\dagger
		\end{array}\right)};
	\end{equation}
	otherwise,
	\begin{equation}
	    \label{eq:pseudo_noninformative}
		\scalemath{0.9}{\mathcal{X}_t^\dagger = \left(\begin{array}{*1{c}}
			\mathcal{X}_{t-1}^\dagger - \epsilon_t \ogamma_t \ozeta_t^\intercal \\ \epsilon_t \ozeta_t^\intercal
		\end{array}\right)},
	\end{equation}
    where $\epsilon_t\in\mathbb{R}$, $\ogamma_t\in\mathbb{R}^t$, and $\ozeta_t\in\mathbb{R}^n$ are defined as,
    \begin{equation}
        \label{eq:fast_params}
        \begin{aligned}
            \epsilon_t \coloneqq \frac{1}{\|\ogamma_t\|^2+1}, \hspace{2mm} \ogamma_t \coloneqq \mathcal{X}_{t-1}^\dagger \x_t, \hspace{2mm}  \ozeta_t  \coloneqq \big( \mathcal{X}_{t-1}^\dagger \big)^\intercal \ogamma_t.
        \end{aligned}
    \end{equation}
\end{proposition}

\begin{proof}
    Rearrange $\mathcal{X}_t$ into,
    \begin{align*}
        \mathcal{X}_t = \left(\begin{array}{*2{c}}
            \mathcal{X}_{t-1} & 0
        \end{array}\right) + \x_t \e_{t+1}^\intercal.
    \end{align*}
    Then, it is implied from Theorem 1 in \cite{Meyer1973generalized} that
    \begin{align*}
		\mathcal{X}_t^{\dagger} = \left(\begin{array}{*2{c}}
		    \mathcal{X}_{t-1} & 0
		\end{array}\right)^\dagger
		+ \left[ \e_{t+1}-\left(\begin{array}{*2{c}}
		    \mathcal{X}_{t-1} & 0
		\end{array}\right)^\dagger \x_t \right] \z_t^\dagger,
	\end{align*}
	whenever $\x_t\not\in\mathcal{R}(\mathcal{X}_{t-1})$.
	Hence, by leveraging the SVD of $\mathcal{X}_{t-1}$ and definition of pseudoinverse we get
	\begin{gather*}
	    \mathcal{X}_t^{\dagger} = \left(\begin{array}{c}
			\mathcal{X}_{t-1}^\dagger \\ 0
		\end{array}\right) + \Big[ \e_{t+1} - \left(\begin{array}{c}
			\mathcal{X}_{t-1}^\dagger \\ 0
		\end{array}\right) \x_t \Big] \z_t^\dagger = \left(\begin{array}{c}
			\mathcal{X}_{t-1}^\dagger - \ogamma_t \z_t^\dagger \\ \z_t^\dagger
		\end{array}\right).
	\end{gather*}
	For the case when $\x_t\in\mathcal{R}(\mathcal{X}_{t-1})$, Theorem 3 in \cite{Meyer1973generalized} gives,
	\begin{align*}
	\scalemath{0.9}{
        \mathcal{X}_t^{\dagger} = \left(\begin{array}{c}
		    \mathcal{X}_{t-1}^\dagger \\ 0
		\end{array}\right) + \e_{t+1} \ogamma_t^\intercal \mathcal{X}_{t-1}^\dagger - \frac{1}{\sigma} \p  \q^\intercal,}
    \end{align*}
	where,
    \(\sigma = \|\ogamma_t\|^2 + 1, \; \p=-\|\ogamma_t\|^2 \e_{t+1} - \scalemath{0.8}{\left(\begin{array}{c}
		    \ogamma_t \\ 0
		\end{array}\right)}, \; \q = -\ozeta_t.\)
    The rest of the proof follows from rearranging the terms and using the identities in \eqref{eq:fast_params}.
\end{proof}
\noindent As mentioned earlier, the update of the controller requires $\mathcal{Y}_{t}\mathcal{X}_{t-1}^{\dagger}$ that could become prohibitive for large $n$.
However, we can take advantage of \Cref{prop:pseudo} to find this term recursively in order to avoid memory usage as well as computational burden.

\begin{theorem}
    \label{thm:recursive_DGR}
    Let $\mathcal{X}_{t-1}$ be as in \Cref{alg:Controller} and $\x_t$ be the state measurement collected at $t$.
    For $t>0$, define $\mathcal{P}_{t-1} \coloneqq \mathcal{X}_{t-1} \mathcal{X}_{t-1}^\dagger$, $\mathcal{Q}_{t-1} \coloneqq \mathcal{Y}_{t} \mathcal{X}_{t-1}^\dagger$ and $\z_t \coloneqq ~\Pi_{\mathcal{R}(\mathcal{X}_{t-1})^\perp} \x_t~$.
    Then
    \begin{align*}
        \z_t &= \big[ \mathrm{I} - \mathcal{P}_{t-1} \big] \x_t,\\
        \mathcal{Q}_t &= \mathcal{Q}_{t-1} + [\x_{t+1} - B\u_t  - \mathcal{Q}_{t-1} \x_t] \z_t^\dagger, \\
        \mathcal{P}_t &= \mathcal{P}_{t-1} + \z_t \z_t^\dagger .
    \end{align*}
\end{theorem}

\begin{proof}
    The expression for $\z_t$ follows directly by properties of pseudoinverse.
    Next, observing that $\x_{t+1} - B\u_t = A \x_t$ and so $\mathcal{Q}_t = A\mathcal{P}_t$, the recursive relation for $\mathcal{Q}_t$ can be derived from the one for $\mathcal{P}_t$.
    Finally, \Cref{prop:pseudo} implies that, if $\x_t\not\in\mathcal{R}(\mathcal{X}_{t-1})$ then
    \begin{gather*}
        \begin{aligned}
            \mathcal{P}_t = \mathcal{X}_t \mathcal{X}_t^\dagger &= \left(\begin{array}{cc}
                \mathcal{X}_{t-1} & \x_t
            \end{array}\right) \left(\begin{array}{c}
    			\mathcal{X}_{t-1}^\dagger - \ogamma_t \z_t^\dagger \\ \z_t^\dagger
    		\end{array}\right) = \mathcal{P}_{t-1} + \z_t \z_t^\dagger;
        \end{aligned}
    \end{gather*}
    otherwise, $\mathcal{P}_{t} = \mathcal{P}_{t-1}$ because $\mathcal{X}_{t-1} \ogamma_t = \mathcal{X}_{t-1} \mathcal{X}_{t-1}^\dagger \x_t = \x_t$.
    But in this case,  $\z_t=(\mathrm{I}-\mathcal{P}_{t-1})\x_t=0$ and therefore, the same recursion holds.
\end{proof}

\noindent Given the recursions introduced in \Cref{thm:recursive_DGR}, the refined (fast) version of \ac{DGR} is displayed in \Cref{alg:Controller_fast}.
At each iteration $t$, we update $\mathcal{Q}_t$ based on the information from the new data and the projection $\mathcal{P}_{t-1}$ (hidden in $\z_t^\dagger$), which itself gets updated as a part of the recursion.
The $n \times n$ matrix $\mathcal{Q}_t$ is then employed for the controller's update.
Notice that $\z_t$ is the same as in \Cref{lem:algo-dyn}, however, here we compute it using $\mathcal{P}_{t-1}$ which is obtained recursively.

\begin{algorithm}[b]
	\caption{\acf{F-DGR}}
	\begin{algorithmic}[1]
		\State \textbf{Initialization}
		\State \hspace{5mm} Measure $\x_0$, set $K_0 = 0$ and {$G_\alpha = [\alpha I + B^\intercal B]^\dagger B^\intercal$}
        \State \hspace{5mm} Set $\mathcal{P}_{-1} = \mathcal{Q}_{-1} = \mathbf{0}$ and $t = 0$
		\State \textbf{While stopping criterion not met}
		\State \hspace{5mm} Compute $\u_t = -K_t \x_t$
		\State \hspace{5mm} Run system \eqref{eqn:sys-dynamic} and measure $\x_{t+1}$
		\State \hspace{5mm} Set \hspace{1mm} $\z_t = [\mathrm{I}-\mathcal{P}_{t-1}]\x_t$
		\State \hspace{12mm} $\mathcal{Q}_t = \mathcal{Q}_{t-1} + [\x_{t+1} - B\u_t - \mathcal{Q}_{t-1} \x_t]\z_t^\dagger$
		\State \hspace{12.5mm} $\mathcal{P}_t = \mathcal{P}_{t-1} + \z_t\z_t^\dagger$
		\State \hspace{12.5mm} $K_{t+1} =  G_\alpha \mathcal{Q}_t$
		\State \hspace{5mm} $t = t+1$
	\end{algorithmic}
	\label{alg:Controller_fast}
\end{algorithm}


{\color{violet}
Next, we discuss the convergence of \ac{DGR} algorithm in the following. Note that, by \Cref{thm:recursive_DGR}, \ac{DGR} and \ac{F-DGR} are equivalent, and the following corollary provides a useful necessary and sufficient conditions for the convergence of $\mathcal{Q}_t$.
\begin{corollary}\label[corollary]{cor:convergence}
    Let $\mathcal{Q}_t$ be as defined in \Cref{alg:Controller_fast} and $\mathcal{X}_{t-1}$ as in \Cref{alg:Controller}. For any $t>0$, if $\x_{t+k} \in \mathcal{R}(\mathcal{X}_{t-1})$ for all $k \geq 0$, or $\mathcal{R}(A^\intercal) \subseteq \mathcal{R}(\mathcal{X}_{t-1})$ then \Cref{alg:Controller_fast} has converged at $t$, i.e., $\mathcal{Q}_{t+k} = \mathcal{Q}_{t-1}$ for all $k \geq 0$.
    Conversely, if \Cref{alg:Controller_fast} has converged at some $t>0$, then for each $k \geq 0$, we must have $\x_{t+k} \in \mathcal{R}(\mathcal{X}_{t+k-1})$ unless $\mathcal{R}(A^\intercal) \subseteq \mathcal{R}(\mathcal{X}_{t+k-1})$.
\end{corollary}
}
{\color{violet}
\begin{proof}
By \Cref{thm:recursive_DGR}, we know that 
\[\mathcal{Q}_t - \mathcal{Q}_{t-1} = A (\mathcal{P}_t - \mathcal{P}_{t-1}) = A \z_t \z_t^\dagger.\]
Therefore, \Cref{alg:Controller_fast} has converged at $t$ if and only if $A \z_{t+k} \z_{t+k}^\dagger = 0$ for all $k \geq 0$. 
And the latter statement is valid if and only if, for each $k \geq 0$, either $\z_{t+k} = 0$ or $\z_{t+k} \in~\mathcal{N}(A)$, which is equivalent to either $\x_{t+k} \in~\mathcal{R}(\mathcal{X}_{t+k-1})$ or $\mathcal{R}(A^\intercal) \subseteq \mathcal{R}(\mathcal{X}_{t+k-1})$.
\end{proof}

A simple implication of the preceding corollary is that \Cref{alg:Controller_fast} converges as soon as the collected data $\mathcal{X}_{t-1}$ is \emph{rich} enough. For instance, in the worst case scenario--when $A \in \mathbb{R}^{n\times n}$ has no zero eigenvalues and all of its modes are excited--the algorithm converges as soon as $n$ linearly independent state measurements have been collected from the noise-less dynamics in \cref{eqn:sys-dynamic}. From this point on, $\u_t$'s as computed in \ac{DGR} and \ac{F-DGR} coincide with the  solution of \cref{eqn:optimization}. However, the convergence of the controller in \ac{DGR} and \ac{F-DGR} may happen earlier in the process whenever the future state measurements lie in the range of previous ones. For example, under hypothesis of \Cref{prop:zero}, when $\x_0$ excites $k$ modes of $A$ corresponding to $r \leq k < n$ distinct eigenvalues, then both \ac{DGR} and \ac{F-DGR} converge in $r$ iterations. In this case, the proposed controller may not coincide with the actual solution of the optimization in \cref{eqn:optimization}. Nonetheless, the online regulation of the system is guaranteed in general by \Cref{thm:recursion}. Note that, in presence of process noise in the dynamics \cref{eqn:sys-dynamic}, \Cref{cor:convergence} is not valid and the convergence behavior of \ac{DGR} (and \ac{F-DGR}) will be dictated by the noise stochastics (see \Cref{fig:convergence} in \Cref{sec:simulation}).
}

Finally, $\z_t$ reflects the informativity of the newly generated data $\x_t$.
In fact, based on its definition, $\mathcal{Q}_t$ provides an estimate of $A$ up to iteration $t$.
Hence, the update of $\mathcal{Q}_t$ as in \Cref{alg:Controller_fast} essentially adjusts the prior estimate of $A$ based on the new information encoded in the term $A \z_t\z_t^\dagger$.
%
%
All in all, the machinery provided in this section circumvents the computational load of finding 
pseudoinverses by leveraging the recursive nature of the solution methodology.

\section{Simulations}
\label{sec:simulation}

    In order to showcase the advantages of the proposed method in practical settings, 
    we have implemented \ac{DGR} on data collected from the X-29A aircraft.
    The Grumman X-29A is an experimental aircraft initially tested for its forward-swept wing;
    it was designed with a high degree of longitudinal static instability (due to the location of the aerodynamic center on the wings) for maneuverability, where linear models were leveraged to determine the closed-loop stability (\Cref{fig:x-29real}).
    The primary task of the control laws is to stabilize the longitudinal motion of the aircraft.
    To this end, the dynamic elements of the flight control system is designed for two general modes: 1) the Normal Digital Powered Approach (ND-PA) used in the takeoff and landing phase of the flight, and 2) the Normal Digital Up-and-Away (ND-UA) when otherwise.

   For both flight modes, we study the case where the dynamics of the aircraft has been perturbed and unknown.
   This can be due to a mis-estimation of system parameters and/or any unpredicted flaw in the flight dynamics due to malfunction/damage.
   In this setting, the control laws designed for the original system fail and the system can become highly unstable.
   We then let \ac{DGR} regulate the system; in this case, since the aircraft continues to operate safely, one can use any data-driven identification, stabilization, or robust control methods once enough data has been collected.
    %

    \begin{figure}[t]
        \centering
        \includegraphics[width=1\columnwidth]{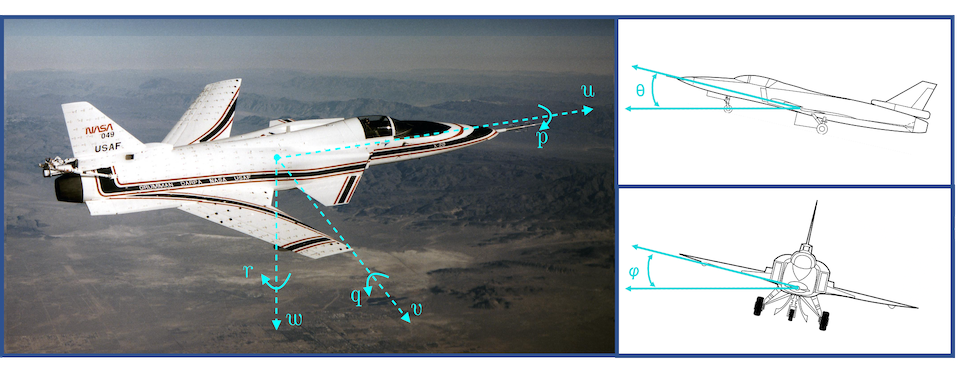}
        \caption{Grumman X-29A \textit{(Credits: NASA Photo)}, mainly known for its extreme instability while providing high-quality maneuverability; the longitudinal and lateral-directional states are illustrated.}
        \label{fig:x-29real}
        \vspace{-0.3cm} 
    \end{figure}

     The longitudinal and lateral-directional dynamics each contains 4 states (see \Cref{fig:x-29real}). The nominal system parameters in each operating mode are obtained from Tables 9-10 and 13-14 in \cite{bosworth1992linearized} (with fixed  discretization step-size $0.05$), whereas perturbation $\Delta A$ is assumed to shift the dynamics to,
    \[\x_{t+1} = (A+\Delta A) \x_t + B \u_t+\oomega_t,\]
    where the elements of $\Delta A$ are sampled from a normal distribution  $\mathcal{N}(0,0.05)$, and $\oomega_t\sim \mathcal{N}(0,0.01)$ denotes the process noise.
   Note that even though the nominal dynamics is known in this example, the proposed machinery makes no such \textit{a priori} estimate, and assumes a completely unknown dynamics $A_{\text{new}} \coloneqq A+\Delta A$.
   {\color{violet}
   The original controller for the unperturbed system in each mode is assumed to be a closed-loop infinite horizon \ac{LQR} with state and input weights $Q=I$ and $R =10^{-7}$.
   }
    
    We now aim to regulate the unstable system $A_{\text{new}}$ from random initial states (where each state is sampled from $\mathcal{N}(0, 10.0)$).
    Note that both the original system and the perturbed system have effective input characteristics that make them  regularizable (with $\rho(\widetilde{A}) = 0.998 $ and $\rho(\widetilde{A}_\text{new}) = 0.927 $ for ND-PA mode, and $\rho(\widetilde{A}) = 0.998 $ and $\rho(\widetilde{A}_{\text{new}}) = 0.932 $ for ND-UA mode).
	The resulting state trajectories for ND-PA and ND-UA modes are demonstrated in Figures~\ref{fig:x-29-ND-PA} and \ref{fig:x-29-ND-UA}, respectively.
	Without \ac{DGR}, the norm of the state $\|\x_t\|$ would grow rapidly (red curve) as the unknown system is unstable and the original control laws fail.\footnote{Since the LQR solution, in general, may have small stability margins for general parameter perturbations \cite{zhou1996robust}.}
	As the plots suggest, with \ac{DGR} in the feedback loop 
	{\color{violet}(with the choice of $\alpha = 5\times 10^{-7}$)},\footnote{The positive choice for $\alpha$ adjusts the compromise between state regulation and reducing the 2-norm of the input. This may lead to a larger upper bound on the state regulation specially when the system is unstable.} the unstable modes can be suppressed resulting in stabilization of the system (the norm of the states in this case is demonstrated in black and each state is depicted in faded color).

	Up to iteration $t=36$ for longitudinal and $t=30$ for lateral directional dynamics  (shown with vertical dashed-line), enough data is generated in order to estimate the new system dynamics, or apply any other data-driven control using the data, (safely) generated by \ac{DGR} up to this point.
	{\color{PineGreen} In what follows, we first showcase the complementary utility of DGR for identification-and-control; we then illustrate how it can also 
	be incorporated for data-driven control.
	}

	\textcolor{violet}{In particular, the data is informative enough to identify system parameters through least squares denoted by $\hat{A}$.
	Therefore, one stopping criterion--which is also used here--is the point where the estimate of system parameters $\hat{A}$ has converged.
	Then, one can replace \ac{DGR} with a closed-loop infinite horizon \ac{LQR} controller with some cost-weights $Q$ and $R$ which is obtained using the new estimate of the system dynamics. Here we set $Q=I$ and $R = 10^{-7}$ in order to make it comparable to the one-step quadratic cost used for \ac{DGR}.
	In contrast to the original unstable \ac{LQR} controller (red curve), it is shown that the new LQR controller for $\hat{A}$ (blue curve) is stabilizing since we now have a more accurate estimate of the (perturbed) system parameters using the data generated safely by \ac{DGR} in the loop.}

	{\color{PineGreen} Next, while the \ac{DGR} is still in effect, the generated data matrix is not ill-conditioned and thus can be utilized to implement a data-driven control algorithm from that point onward.
	Due to the presence of noise and uncertainty, we have implemented the regularized version of Data-driven \ac{MPC} as in \cite{coulson2019deepc, berberich2021data} with parameters $T_{\text{ini}} = 1$, $N=4$, $Q = 400 I$, $R = 0.05 I$, $\lambda_\sigma =10^4$ and $\lambda_g =1$ for both dynamics, where the input is persistently exciting.
	With \ac{DGR}, after enough data has been generated for each dynamics, the data-driven \ac{MPC} algorithm 
	is initiated; the norm of the corresponding state vector is depicted in yellow dash-dotted line labeled as ``\ac{DGR}+DeePC.''

	On the other hand, one could consider implementing the data-driven \ac{MPC} without DGR. {\color{PineGreen} However, this would require offline data which is not available a priori. Nonetheless, just for the purpose of comparison, this has been implemented based on {\em offline data} obtained from the original unstable plant. The resulting norm of the state has been depicted in cyan labeled as ``Offline data+DeePC''. Due to the ill-posedness of the data matrix resulting from an unstable plant, it is observed that the practical tuning of the parameters could be problematic. This is due to the fact that  maintaining the stability and feasibility of the resulting convex optimization problems is challenging due to conditioning in the dataset (for similar observations see for example~\cite{de2019formulas}). These trajectories are terminated whenever the corresponding optimization problem was not numerically solvable/feasible. We have used the CVXPY package for solving the convex programs derived in all these cases \cite{diamond2016cvxpy}.
	}}

    \begin{figure}[!t]
\centering
\subfloat[]{
    \hspace{-1cm}\includegraphics[trim=-1cm 0 0 2.7cm, clip, width=\columnwidth]{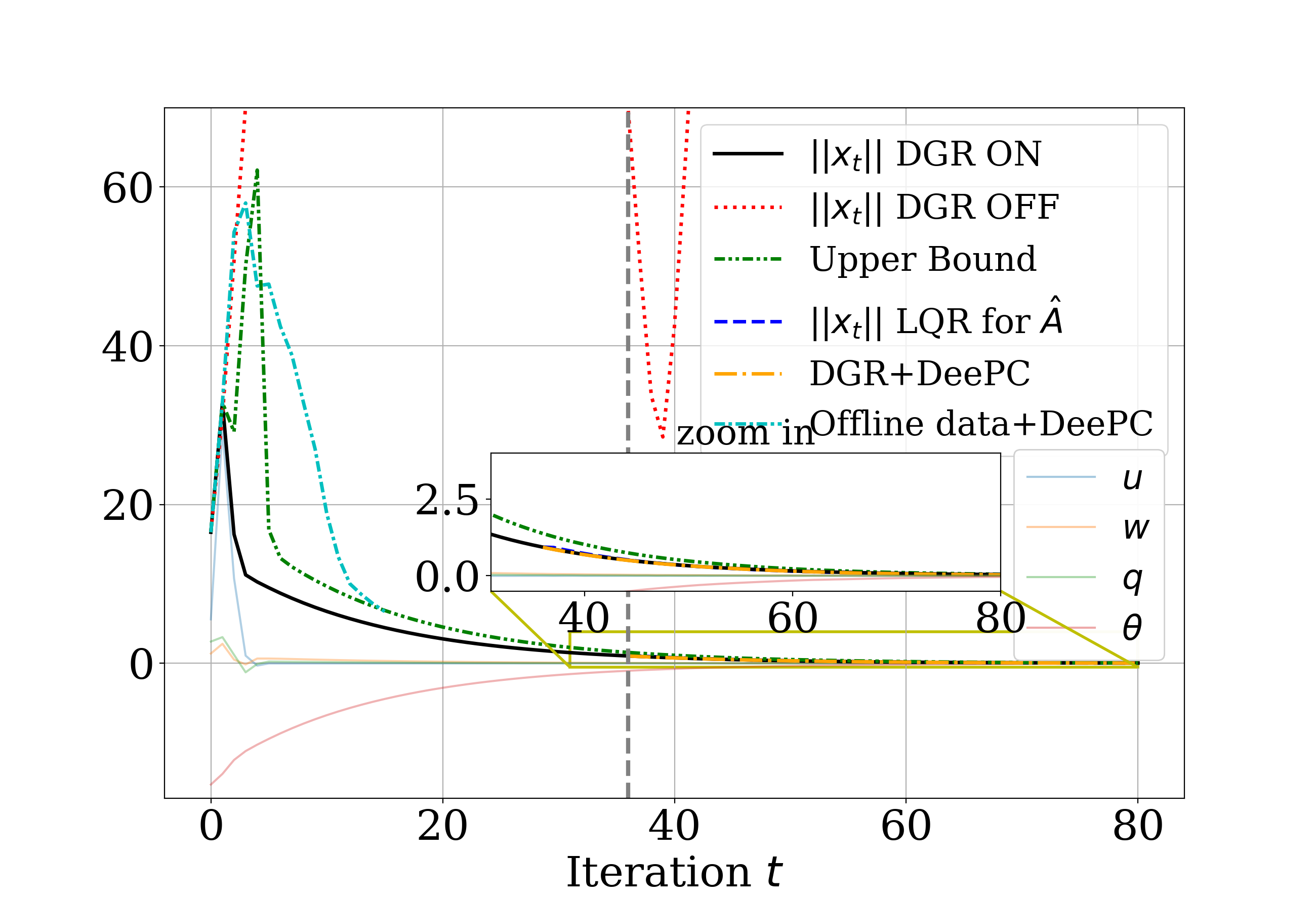}
    \label{fig:x-29-ND-PA-long}
}
\hfil
\subfloat[]{
    \hspace{-1cm}\includegraphics[trim=-1cm 0 0 2.7cm, clip, width=\columnwidth]{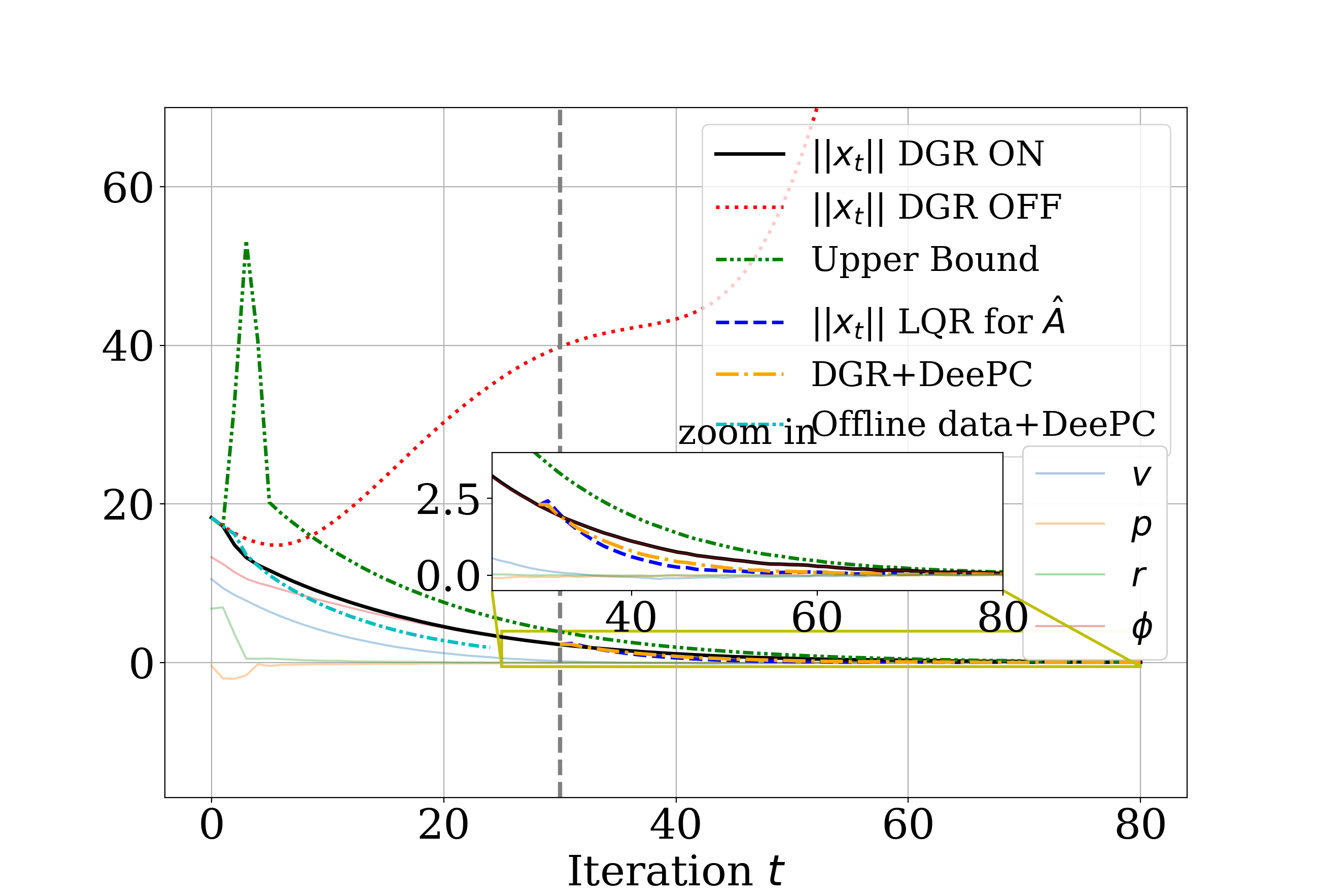}
    \label{fig:x-29-ND-PA-lateral}
}
\caption{The state trajectory of X-29 in ND-PA mode with and without \ac{DGR} for a)  longitudinal control, b) lateral-directional control.}
\label{fig:x-29-ND-PA}
\vspace{-0.3cm}
\end{figure}

	


\begin{figure}[!t]
\centering
\subfloat[]{
    \hspace{-1cm}\includegraphics[trim=-1cm 0 0 2.7cm, clip, width=\columnwidth]{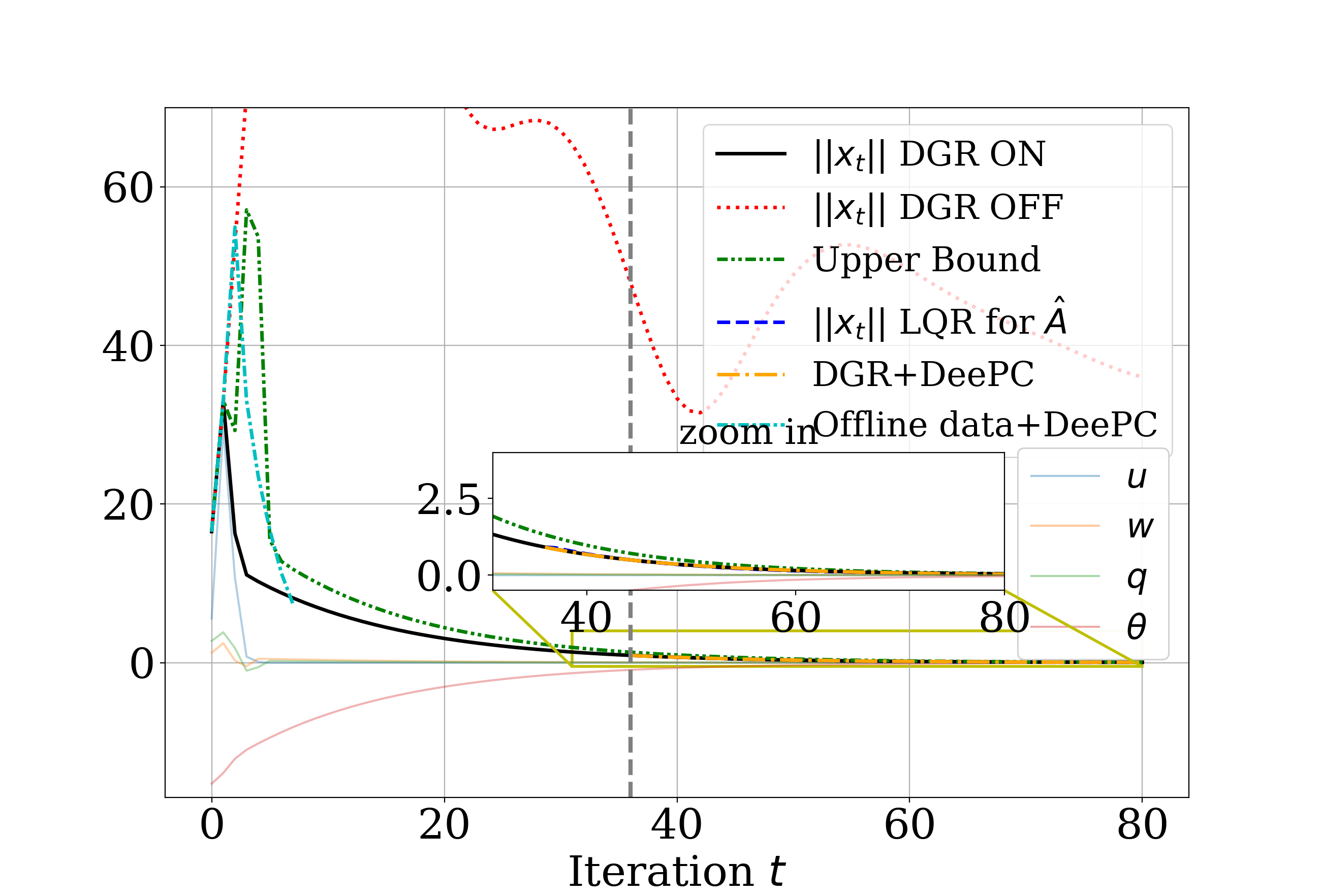}
    \label{fig:x-29-ND-UA-long}
}
\hfil
\subfloat[]{
    \hspace{-1cm}\includegraphics[trim=-1cm 0 0 2.7cm, clip, width=\columnwidth]{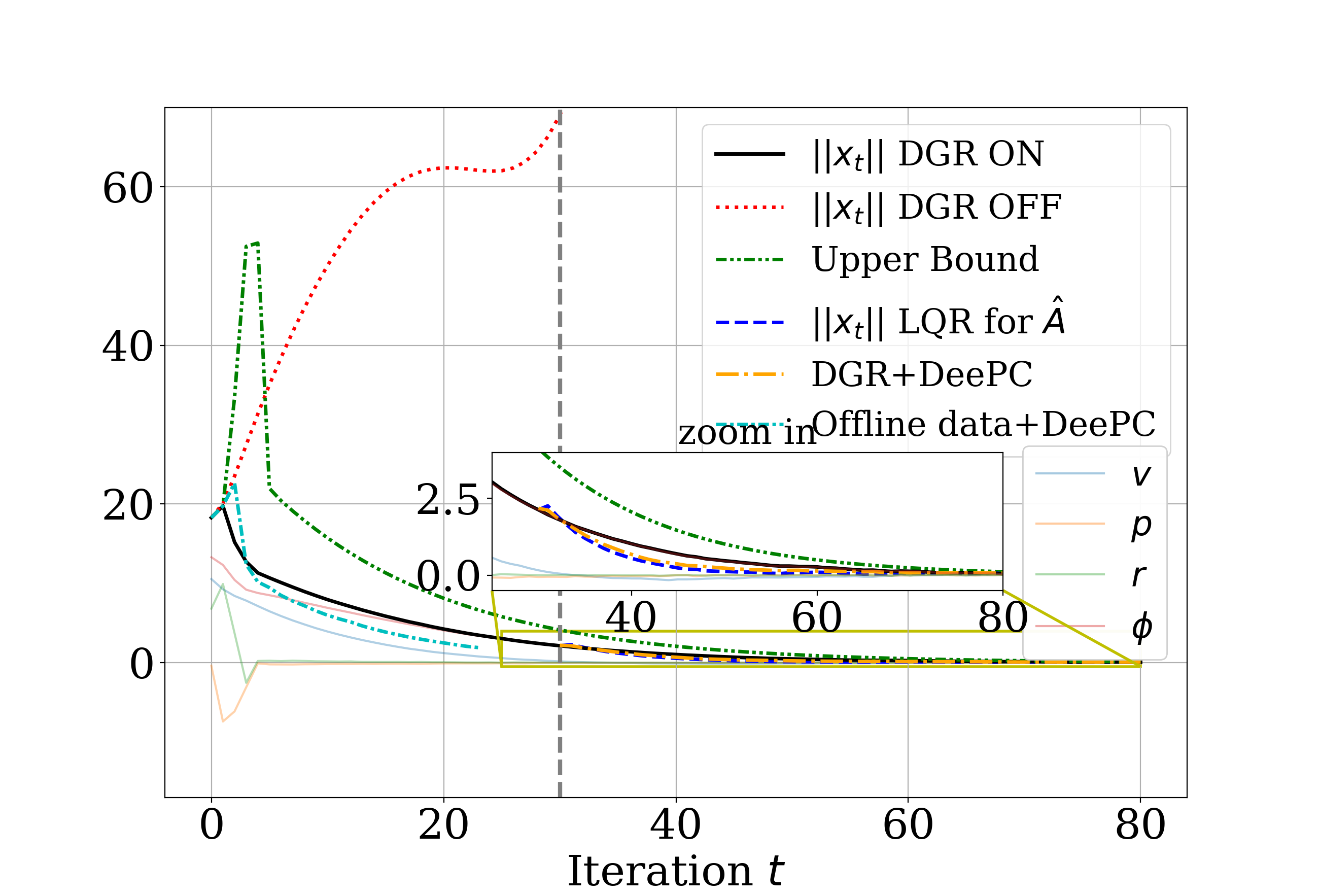}
    \label{fig:x-29-ND-UA-lateral}
}
\caption{\color{violet}The state trajectory of X-29 in ND-UA mode with and without \ac{DGR} for a)  longitudinal control, b) lateral-directional control.}
\label{fig:x-29-ND-UA}
\end{figure}



	\textcolor{violet}{Finally, the convergence of DGR algorithm in terms of the designed controller $K_t$ is illustrated in \Cref{fig:convergence}, where the values are seen to be well behaved after 30 to 40 iterations from the noisy dynamics.
	Furthermore, we note that as the process noise $\oomega_t$ decreases, the error for $K_t$ tends to zero; in the absence of noise, the limiting error is in fact negligible.}
	
	\begin{figure}[!t]
    \centering
    \includegraphics[trim= 0 0 0 0, clip, width=0.9\columnwidth]{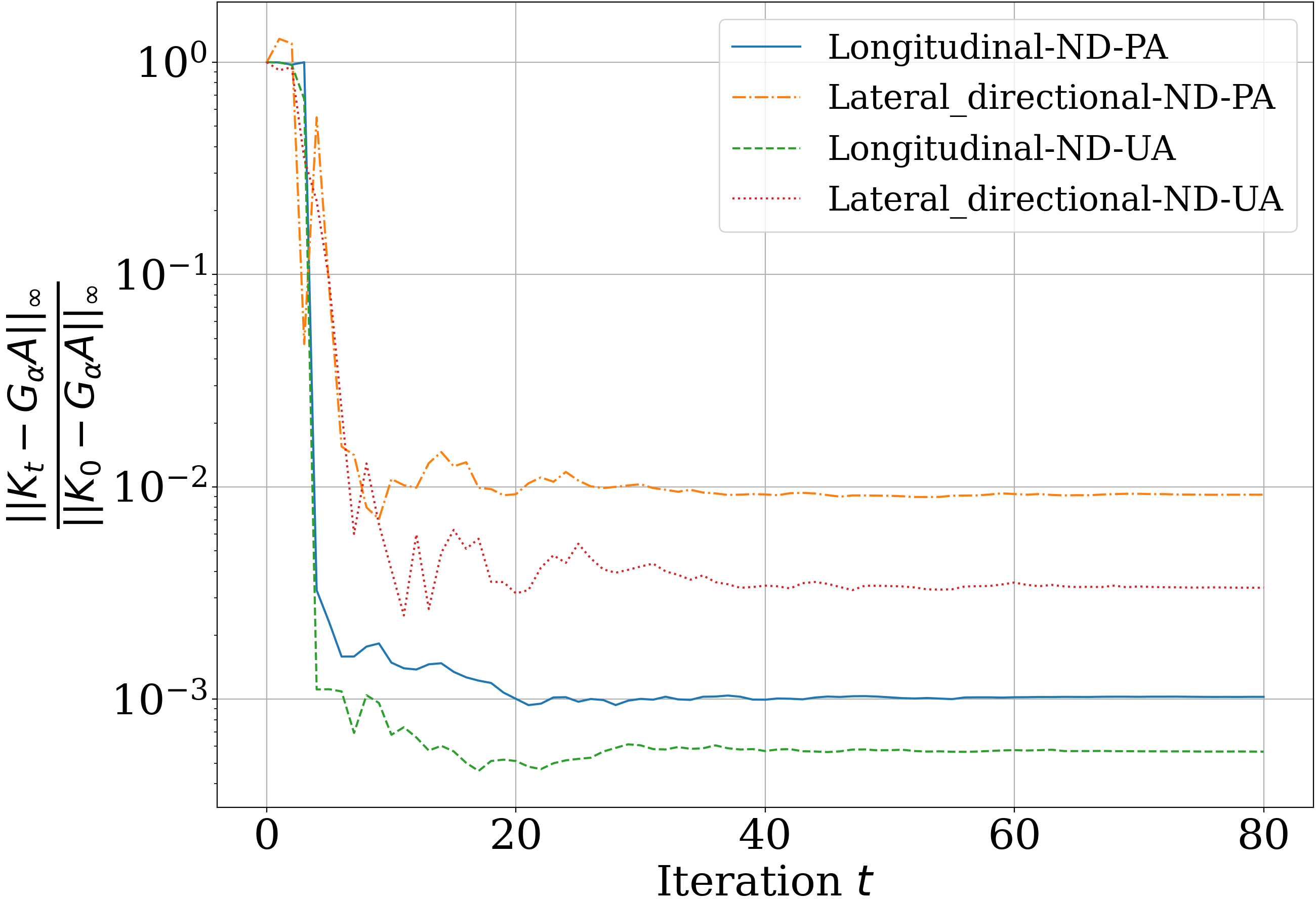}
    \caption{The convergence behavior of \ac{DGR} algorithm for longitudinal and lateral-directional dynamics for both of flight modes.}
    \label{fig:convergence}
    \vspace{-0.3cm}
    \end{figure}
    
	{\color{Blue}
	In these examples, the bound derived in \Cref{thm:recursion}--that requires the knowledge of $A_{\text{new}}$-- is plotted in green for comparison.}
	The behavior of the bound follows our observations in \Cref{rmk:bound_behavior};
	the bound increases as the algorithm initially tries to ``detect" the unstable modes, followed by suppressing these modes for regulation.
	We finally note that for large enough iterations, the rate of change of the upper bound is dictated by $\rho(\widetilde{A}_\text{new})$, which in this case, is slightly less than one.\footnote{The code for these simulations can be found at \url{https://github.com/shahriarta/Data-Guided-Regulation}.}
	

\section{Conclusion}
\label{sec:conclusion}
{    In this paper, we have introduced and characterized ``regularizability,'' a system theoretic notion to quantify the ability of data-driven finite-time regulation for a linear system. Regularizability is distinct from constructs that characterize asymptotic behavior of the system, such as stabilizability and controllability.
    Furthermore, we have proposed \ac{DGR}, an online iterative feedback regulator for a partially unknown, potentially unstable, linear system using streaming data from a single trajectory.
	In addition to regulation of an unknown system, \ac{DGR} leads to informative data that can subsequently be used 
	for data-guided stabilization or system identification.\footnote{
	{\color{violet} This is guaranteed for example when $\alpha = 0$.}}
    Along the way, we have provided another system theoretic notion referred to as ``instability number'' in order to analyze the performance of DGR and derive bounds on the trajectory of the system over a finite-time interval.}
	Subsequently, we presented the application of the proposed online synthesis procedure on a highly maneuverable unstable aircraft. This example underscores {\color{PineGreen} how \ac{DGR} can be integrated with other state-of-the-art data driven methods to achieve better performance through improved numerical conditioning}.

	The extensions of the results presented in this paper to noisy dynamics as well as considering an unknown input matrix---in price of the guaranteed performance from the onset---are deferred to our future work.
	Furthermore, state regulation becomes more challenging when one only relies on partial observation of system's trajectory or when the system is known to have multi-scale dynamics.
	%
	Finally, our setup would be more practical considering input constraints, e.g.,
	rate limits.
	While it is straightforward to address such extensions via convex constraints in the proposed DGR procedure, analysis of the resulting closed loop trajectory
	is more involved.

	
\vspace{-0.5cm}
\appendix{\em Proof of \Cref{lem:Mbound-gen}}:
\label{app:proof-Mbound}
		Let $A = W \Sigma U^\intercal$ be the \ac{SVD} of $A$ where $\Sigma$ is diagonal containing the singular values in descending order and both $W,U \in \mathbb{R}^{n\times n}$ are unitary.
		This implies that,
		\begin{align*}
		M_t(A) &\textstyle = \sup_{\{\v_i\}_1^t \in \mathcal{O}_t^n} \hspace{2mm} \|\Sigma U^\intercal \v_1\| \, \|\Sigma U^\intercal \v_2\| \, \hspace{1mm}\cdots \hspace{1mm}  \|\Sigma U^\intercal \v_t\| \\
		&\textstyle = \sup_{\{\v_i\}_1^t \in \mathcal{O}_t^n} \hspace{2mm} \|\Sigma \v_1\| \, \|\Sigma \v_2\| \, \hspace{1mm}\cdots \hspace{1mm}  \|\Sigma \v_t\| ,
		\end{align*}
		where the last equality is due to the fact that $\{U^\intercal \v_i\}_1^t \in \mathcal{O}_t^n$ only if $\{\v_i\}_1^t \in \mathcal{O}_t^n$, since $U$ is unitary.
		For the lower-bound, if $t \leq n$, we can choose $\{\v_i\}_1^t \in \mathcal{O}_t^n$ such that $|\langle \e_1, \v_i \rangle| = 1/\sqrt{t}$ for all $i=1,\cdots,t$.
		This choice is certainly possible as a result of applying Parseval's identity in a $t$-dimensional subspace with orthonormal basis $\{\v_i\}_1^t$ containing the unit vector $\e_1$, in which, $\e_1$ is represented with all coordinates equal to $1/\sqrt{t}$ 
		with respect to this basis.
		We thus conclude that,
		\begin{equation*}
		\scalemath{0.9}{
		M_t(A) \geq |\sigma_1 \langle \e_1, \v_1 \rangle| \, \cdots \, |\sigma_1 \langle \e_1, \v_t \rangle| \geq \big( {\sigma_1}/{\sqrt{t}} \big)^t ,}
		\end{equation*}
		where the left inequality follows from the fact that $\|\Sigma \v\|~\geq~ |\sigma_1 \langle \e_1, \v \rangle|$ for any $\v \in \mathbb{R}^n$.
		For the upper-bound, define $\Sigma_t = \mathrm{diag}(\sigma_1, \dots, \sigma_t)$ and since singular values are in descending order we have,
		\begin{gather*}
		\begin{aligned}
		M_t(A)& \textstyle \leq \sup_{\{\v_i\}_1^t\in \mathcal{O}_t^t} \prod_{i=1}^t \|\Sigma_t \v_i \| \\
		&\textstyle = \sup_{\{\v_i\}_1^t\in \mathcal{O}_t^t} \prod_{i=1}^t \Big[ \sigma_1^2\left|\langle \e_1, \v_i \rangle\right|^2 + \sum_{j=2}^t \left|\sigma_j \langle \e_j, \v_i \rangle\right|^2 \Big]^{\frac{1}{2}}\\
		& \textstyle \leq \sup_{\{\v_i\}_1^t\in \mathcal{O}_t^t} \prod_{i=1}^t \Big[\sigma_1^2\left|\langle \e_1, \v_i \rangle\right|^2 + \delta \Big]^{\frac{1}{2}} .
		\end{aligned}
		\end{gather*}
		Define $\gamma_i =  \langle \e_1, \v_i \rangle$; then by Bessel's inequality $\sum_{i=1}^t \gamma_i^2 \leq 1$ whenever $\{\v_i\}_1^t\in \mathcal{O}_t^t$. Thereby, by denoting $\ogamma \coloneqq [
		\gamma_1 \;
		\hdots \; \gamma_t
		]^\intercal$, we can conclude that
		\begin{gather*}
		\begin{aligned}
		& M_t(A) \leq \textstyle \sup_{\ogamma \in \mathcal{B}_{2}^t} \prod_{i=1}^t \Big[ \sigma_1^2 \gamma_i^2 + \delta \Big]^{\frac{1}{2}}\\
		&\textstyle = \sup_{\ogamma \in \mathcal{B}_{2}^t} \Big[ \sum_{i=1}^{t+1} \sigma_1^{2(t+1-i)} \, \delta^{i-1} \hspace{-1mm} \sum_{\substack{|\alpha| = t+1-i\\ \alpha \in \{0,1\}^t}~} (\gamma_1^2)^{\alpha_1} \cdots (\gamma_t^2)^{\alpha_t} \Big]^{\frac{1}{2}},
		\end{aligned}
		\end{gather*}
		where $\alpha$ is a multi-index of dimension $t$, and the last equality follows by direct computation.
		Now it is easy to see that for a fixed multi-index $\alpha$, if $|\alpha| =m > 0$ and $\alpha \in \{0,1\}^t$ then
		\begin{align*}
		\textstyle
		\sup_{\ogamma \in \mathcal{B}_{2}^t}~ (\gamma_1^2)^{\alpha_1} \cdots (\gamma_t^2)^{\alpha_t} \leq (\frac{1}{m})^m ,
		\end{align*}
        which follows by the symmetry in optimization variables.
		Therefore, we can conclude that
		\begin{gather*}
		\textstyle 
		M_t(A) \leq \Big[ \delta^t + \sum_{i=1}^{t} \big[\frac{\sigma_1^2}{t+1-i}\big]^{(t+1-i)} \, \delta^{i-1} \binom{t}{t+1-i} \Big]^{\frac{1}{2}},
		\end{gather*}
		implying the claimed upperbound.



\vspace{-0.3cm}
\section*{Acknowledgment}
The authors would like to thank Dillon R. Foight and Bijan Barzgaran for valuable discussions pertaining to this work.
The authors also thank the Associate Editor and the anonymous reviewers for their helpful comments on earlier drafts of the manuscript.



\bibliographystyle{ieeetr}
\bibliography{citations}
%
\begin{IEEEbiography}[{\includegraphics[width=1in,height=1.25in,clip,keepaspectratio]{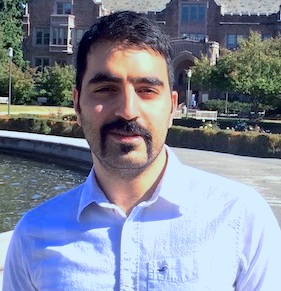}}]{Shahriar Talebi}
    
    (S'17) received his B.Sc. degree in Electrical Engineering from Sharif University of Technology, Tehran, Iran, in 2014, M.Sc. degree in Electrical Engineering from University of Central Florida (UCF), Orlando, FL, in 2017, both in the area of control theory.
    He is currently pursuing a Ph.D. degree in Aeronautics and Astronautics and a M.S. degree in Mathematics at the University of Washington (UW), Seattle, WA. 
    He is a recipient of William E. Boeing Endowed Fellowship, Paul A. Carlstedt Endowment, and Latvian Arctic Pilot--A. Vagners Memorial Scholarship, in 2018-19 at UW, and Frank Hubbard Engineering Scholarship in 2017 at UCF.
    
    His research interests includes game theory and variational analysis, data-driven control and networked dynamical systems. 
    He is also interested in applying the developed techniques to analyze distributed systems operating in cooperative/non-cooperative environments.
\end{IEEEbiography}
\begin{IEEEbiography}[{\includegraphics[width=1in,height=1.25in,clip,keepaspectratio]{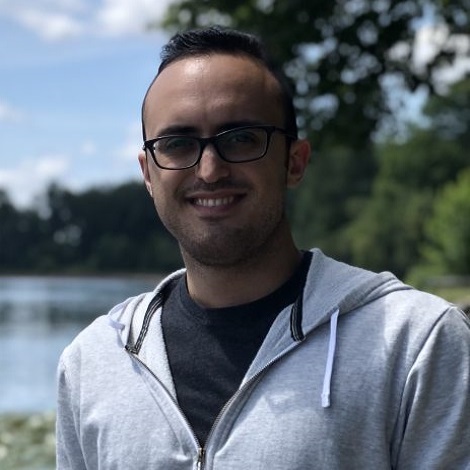}}]{Siavash Alemzadeh}
    (S'17) received his B.S. in Mechanical Engineering from Sharif University of Technology, Tehran, Iran in 2014. He received his M.S. in Mechanical Engineering in 2015 and his Ph.D. from William E. Boeing department of Aeronautics and Astronautics in 2020 from the University of Washington, Seattle, USA. He is currently a data and applied scientist at Microsoft.
    
    His research interests are reinforcement learning, data-driven control of distributed systems, and networked dynamical systems and he is interested in the applications of these areas to transportation, infrastructure networks, and robotics.
\end{IEEEbiography}
\begin{IEEEbiography}[{\includegraphics[width=1in,height=1.25in,clip,keepaspectratio]{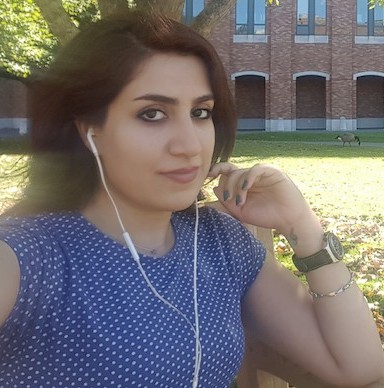}}]{Niyousha Rahimi} (S'20)
    received her B.S. in Mechanical Engineering from Sharif University of Technology, Tehran, Iran in 2016.
    She received her M.S. in Mechanical Engineering from the University of Washington (UW) in 2018.
    She is currently pursuing a Ph.D. in Aeronautics and Astronautics at the University of Washington, WA, USA. She was a recipient of 
    Ruth C. Hertzberg Endowed Fellowship in 2018.
    
    Her research interests are data-driven control, Vision-based Navigation and Stochastic planning. She is also interested in the applications of these areas to robotics and aerospace systems.
\end{IEEEbiography}
\begin{IEEEbiography}[{\includegraphics[width=1in,height=1.25in,clip,keepaspectratio]{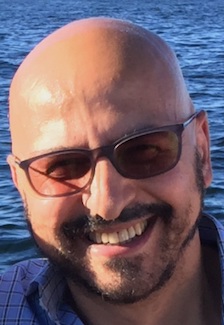}}]{Mehran Mesbahi}
    (F'15) received his Ph.D. from the University of Southern California, Los Angeles, in 1996. He was a member of the Guidance, Navigation, and Analysis group at JPL from 1996-2000 and an Assistant Professor of Aerospace Engineering and Mechanics at the University of Minnesota from 2000-2002.
    He is currently a Professor of Aeronautics and Astronautics, Adjunct Professor of Electrical and Computer Engineering and Mathematics, and Executive Director of Joint Center for Aerospace Technology Innovation at the
    University of Washington.
    He was the recipient of NSF CAREER Award in 2001, NASA Space Act Award in 2004, UW Distinguished Teaching Award in 2005, and UW College of Engineering Innovator Award for Teaching in 2008.
    
    His research interests are distributed and networked aerospace systems, systems and control theory, and learning.
\end{IEEEbiography}



\end{document}